\documentclass[11pt]{article}

\usepackage[T1]{fontenc}
\usepackage[utf8]{inputenc}
\usepackage[margin=1in]{geometry}
\usepackage{amsmath,amssymb,amsthm,mathtools}
\usepackage{bm}
\usepackage[hidelinks]{hyperref}
\usepackage{enumitem}
\usepackage{array}

\theoremstyle{plain}
\newtheorem{theorem}{Theorem}[section]
\newtheorem{lemma}[theorem]{Lemma}
\newtheorem{proposition}[theorem]{Proposition}
\newtheorem{corollary}[theorem]{Corollary}
\newtheorem{conjecture}[theorem]{Conjecture}

\theoremstyle{definition}
\newtheorem{definition}[theorem]{Definition}
\newtheorem{assumption}[theorem]{Assumption}
\newtheorem{problem}[theorem]{Problem}

\theoremstyle{remark}

\newcommand{\R}{\mathbb{R}}
\newcommand{\C}{\mathbb{C}}
\newcommand{\Z}{\mathbb{Z}}

\newcommand{\E}{\mathbb{E}}
\newcommand{\Sone}{\mathbb{S}^1}
\newcommand{\Stwo}{\mathbb{S}^2}
\newcommand{\dd}{\mathrm{d}}
\newcommand{\Var}{\operatorname{Var}}
\newcommand{\Cov}{\operatorname{Cov}}

\newcommand{\Sp}{\operatorname{Sp}}
\newcommand{\SL}{\operatorname{SL}}
\newcommand{\GL}{\operatorname{GL}}
\newcommand{\Diff}{\operatorname{Diff}}
\newcommand{\KL}{\operatorname{KL}}
\newcommand{\vM}{\mathrm{vM}}
\newcommand{\ent}{\mathsf{S}}          
\newcommand{\Fish}{\mathcal{I}}        
\newcommand{\Leb}{\lambda}             
\newcommand{\kimephase}{\theta}        
\newcommand{\muI}{\mu}                 
\newcommand{\meandir}{\bar\theta}      

\title{\textbf{Kime-Representation Formulations of Three Open Problems
in the Foundations of Classical Mechanics:\\[2pt]
Uncertainty, Invariant Entropy, and Directional Degrees of Freedom}}
\author{Ivo D. Dinov\\
Statistics Online Computational Resource (SOCR)\\
University of Michigan, Ann Arbor, MI 48109, US\thanks{SOCR Spacekime
Group, University of Michigan. The formulations below are stated
relative to the open-problem chapter of the \emph{Assumptions of
Physics} project \cite{CarcassiAidalaBook,CarcassiAidalaEntropy}
(Problems~1.16, 1.20, and 1.21 in the numbering of that draft) and to
the kime manuscripts \cite{KPTms,WDWms,DiracKime,DinovVelevBook}.}}
\date{\today}

\begin{document}
\maketitle

\begin{abstract}
We give mathematically self-contained formulations, in the
complex-time (\emph{kime}) representation, of three open problems from
the foundations of classical mechanics: (I)~the extension of the
classical entropic uncertainty principle to non-canonical variables
and to multiple degrees of freedom; (II)~the characterization of
coordinate-invariant measures and entropies, i.e., the question of why
continuous physical quantities must be \emph{paired} for an invariant
entropy to exist; and (III)~the construction of a classical
relativistic directional degree of freedom (a classical analogue of a
spin-$\tfrac12$ system). Throughout, the kime phase
$\kimephase\in\Sone$ is interpreted \emph{statistically} as a
latent circular random variable whose law $\Phi(\kimephase\mid t)$
models the intrinsic trial-to-trial variability of repeated,
identically controlled experiments indexed by the kime magnitude
$t=|\kappa|$, $\kappa=t e^{i\kimephase}$. The mathematical bridge is
an exact symplectic identification of the kime cone with the
action--angle chart of a one-degree-of-freedom phase space, under
which the kime measure $t\,\dd t\,\dd\kimephase$ is the Liouville
measure and the phase law becomes the angular conditional of a
Liouville density. Within this dictionary we (i)~prove a sharp
entropic uncertainty relation on the cylinder $\Sone\times\R$ whose
extremal family is von Mises~$\otimes$~Gaussian, together with a sharp
circular Fisher-information inequality saturated exactly by von Mises
laws; (ii)~prove an exact non-canonical uncertainty relation in which
the correction term is the \emph{geometric mean} of the Poisson
bracket, clarifying the conjectured role of the expected bracket;
(iii)~prove aggregate multi-degree-of-freedom bounds via the
Williamson normal form and Fischer's inequality, and isolate the
per-degree-of-freedom refinement as a precise open problem of
symplectic Schur--Horn type; (iv)~prove that diffusion of the kime
phase produces monotone entropy growth with the equipartitioned
(Haar-uniform) phase law as the unique attractor, giving rigorous
content to the ``equipartition of entropy'' conjecture; (v)~prove that
a diffeomorphism-covariant theory of continuous quantities admits an
invariant entropy if and only if quantities are canonically paired,
with the Liouville entropy unique up to an additive constant, and
exhibit the kime chart as the K\"ahler normal form of the resulting
pairing; and (vi)~formulate the relativistic directional problem on
Poincar\'e coadjoint orbits, prove the fibered circular uncertainty
relation in the nonrelativistic sector, and translate the
``four-vector versus two-form'' dichotomy into a precise moment-map
criterion informed by the absence of chirality in the
$\mathrm{Cl}(3,2)$ kime compactification. Open problems and
conjectures are stated in a form directly addressable by the
kime-phase tomography inference framework.
\end{abstract}


\section{Introduction}\label{sec:intro}

The Carcassi and Aidala \emph{Assumptions of Physics} project \cite{CarcassiAidalaBook}
derives classical Hamiltonian mechanics from informational premises:
states are identified with distributions over a continuum of
possibilities, the count of states must be independent of the
coordinates used to label them, and deterministic-and-reversible
evolution must preserve that count. Within this program the classical
uncertainty principle appears as an \emph{entropic} statement
\cite{CarcassiAidalaEntropy}.
As Hamiltonian evolution preserves
the Liouville measure, the differential entropy of a state is an
invariant, and the product of marginal uncertainties of a canonical
pair is bounded below by a function of that invariant. The specific three open
problems explored in ti study are described below.

\begin{enumerate}[label=(\Roman*), leftmargin=2.5em]
\item \textbf{(Problem 1.16, uncertainty.)} Extending the classical
uncertainty principle from canonical pairs of a single degree of
freedom (DOF) to (a)~non-canonical variable pairs, with the minimum
uncertainty conjecturally governed by the (expectation of the) Poisson
bracket, and (b)~multiple DOF, where one asks how uncertainty and
correlation migrate between DOF under the symplectic group, and
whether ``equipartition of entropy over uncorrelated DOF'' is a lower
bound.
\item \textbf{(Problem 1.21, invariant entropy.)} Explaining, more
generally than the classical derivation, why continuous quantities
must come in \emph{pairs} for a coordinate-invariant entropy to exist,
with suggested connections to measures on the complex plane and to
generalized complex structures.
\item \textbf{(Problem 1.20, directional DOF.)} Constructing a
classical \emph{relativistic} directional degree of freedom, a
classical analogue of a spin-$\tfrac12$ system, including the
identification of the correct phase space, the correct conjugate
variables generalizing $\{\theta^{xy},S_z\}=1$, and the resolution of
whether spin generalizes to a four-vector or to a two-form.
\end{enumerate}

The \emph{kime} (complex-time) representation
\cite{DinovVelevBook,KPTms,WDWms,DiracKime} replaces the ordering
variable $t$ of repeated experiments by a complex coordinate
$\kappa=t e^{i\kimephase}$ on the time cone
$\mathcal M=[0,T]\times\Sone$, where the magnitude $t$ orders
observations and the phase $\kimephase$ is a latent circular variable.
The present paper adopts throughout the \emph{statistical}
interpretation of the phase, which we fix as assumption
\ref{ass:interp}.

\begin{assumption}[Statistical interpretation of the kime
phase]\label{ass:interp}
The kime phase $\kimephase$ is not a directly controllable or directly
observable coordinate. It is a latent random variable on $\Sone$
whose conditional law $\Phi(\kimephase\mid t)$ models the
\emph{intrinsic domain variability} exhibited by repeated measurements
of the same controlled experiment at clock reading~$t$: independent
repetitions $j=1,\dots,N$ of the experiment correspond to independent
draws $\Theta_j(t)\sim\Phi(\cdot\mid t)$, and observables are
functions (possibly noisy) of $(t,\Theta_j(t))$, as in the
kime-phase-tomography (KPT) observation model
$Y_{j,k}=\mathcal S(t_k,\Theta_j(t_k))+\varepsilon_{j,k}$ of
\cite{KPTms}. All theorems below are statements about this
representation; no claim is made that $\kimephase$ is an ontic
mechanical coordinate. Where a mechanical reading is used (the
action--angle dictionary of Section~\ref{sec:foundations}), it is
introduced as an explicit, falsifiable modeling identification.
\end{assumption}

The contribution of this white paper is to show that, under
Assumption~\ref{ass:interp} plus one exact symplectic identification
(Lemma~\ref{lem:actionangle}, the kime cone with its canonical measure
\emph{is} the action--angle chart of a one-DOF phase space with its
Liouville measure), the three open problems (I)--(III) acquire
kime-native formulations where
(a) a nontrivial portion of each problem becomes a \emph{theorem}
provable with the circular-statistics and information-geometric tools
already developed in \cite{KPTms,WDWms,DiracKime}; and
(b) the genuinely open remainder becomes a sharply stated problem or
conjecture, expressed in terms of objects (phase laws, trigonometric
moments, circular Fisher information, symplectic spectra) that are
\emph{estimable from repeated-measurement data} by kime-phase
tomography, so that partial numerical evidence is obtainable in
principle.

Section~\ref{sec:foundations} fixes the kime-representation
foundations. Section~\ref{sec:uncertainty} treats Problem~(I),
Section~\ref{sec:invariance} treats Problem~(II), and
Section~\ref{sec:directional} treats Problem~(III).
Section~\ref{sec:conclusions} collects interpretations and
conclusions. All proofs are given in full except where a result is
classical and explicitly cited.

\paragraph{Notational conventions.}
$\Sone=\R/2\pi\Z$ is parametrized by $\kimephase\in[-\pi,\pi)$;
$\dd\kimephase$ denotes Lebesgue measure on $\Sone$ (total mass
$2\pi$) and $\dd\kimephase/2\pi$ the Haar probability measure, the
normalization used in \cite{KPTms}. Densities on $\Sone$ are taken
with respect to $\dd\kimephase$ unless stated otherwise; the
conversion to the Haar convention multiplies densities by $2\pi$ and
shifts entropies by $\log 2\pi$, and we indicate this wherever both
conventions appear. For a probability density $\rho$ with respect to
a reference measure $\Leb$ on a measurable space $X$, the
(differential) entropy relative to $\Leb$ is
\[
  \ent_{\Leb}[\rho] \;=\; -\int_X \rho\,\log\rho\;\dd\Leb ,
\]
whenever the integral is well defined in $[-\infty,+\infty)$; the
subscript is dropped when the reference measure is clear
\cite{CoverThomas}. $\KL(\cdot\|\cdot)$ and
$\chi^2(\cdot\|\cdot)$ denote the Kullback--Leibler and chi-squared
divergences. Also, all logarithms are natural, $\Sp(2n,\R)$ is the real
symplectic group,
$\Omega=\begin{psmallmatrix}0&I_n\\-I_n&0\end{psmallmatrix}$ the
standard symplectic form in coordinates
$z=(q^1,\dots,q^n,p_1,\dots,p_n)$, and for smooth $f,g$,
$\{f,g\}=\sum_i(\partial_{q^i}f\,\partial_{p_i}g-
\partial_{p_i}f\,\partial_{q^i}g)$.

\section{Kime-representation foundations}\label{sec:foundations}

\subsection{The kime cone and the phase law}

\begin{definition}[Kime coordinate and time cone]\label{def:kime}
The \emph{kime coordinate} is $\kappa=t e^{i\kimephase}\in\C$ with
\emph{kime magnitude} $t=|\kappa|\ge 0$ and \emph{kime phase}
$\kimephase\in\Sone$. The \emph{time cone} is the manifold-with-apex
$\mathcal M=[0,T]\times\Sone$ (apex $t=0$), equipped with the cone
metric and canonical measure
\begin{equation}\label{eq:conemetric}
  g_0=\dd t^2+t^2\,\dd\kimephase^2,
  \qquad
  \dd\mu_{g_0}= t\,\dd t\otimes\frac{\dd\kimephase}{2\pi},
\end{equation}
as in \cite{KPTms}.
\end{definition}

\begin{definition}[Phase law]\label{def:phaselaw}
A \emph{phase law} is a measurable family
$\{\Phi(\cdot\mid t)\}_{t\in[0,T]}$ of probability densities on
$\Sone$ with respect to $\dd\kimephase$:
$\Phi(\cdot\mid t)\ge 0$ and
$\int_{-\pi}^{\pi}\Phi(\kimephase\mid t)\,\dd\kimephase=1$ for each
$t$. Its \emph{trigonometric moments} are
$\alpha_n(t)=\E[e^{in\Theta_t}]
=\int_{-\pi}^{\pi}e^{in\kimephase}\Phi(\kimephase\mid t)\,
\dd\kimephase$, $n\in\Z$. The \emph{mean resultant length} is
$r(t)=|\alpha_1(t)|\in[0,1]$, and when $r(t)>0$ the \emph{mean
direction} $\meandir(t)$ is defined by
$\alpha_1(t)=r(t)e^{i\meandir(t)}$. The \emph{circular variance} is
$V(t)=1-r(t)$ \cite{MardiaJupp}.
\end{definition}

Under Assumption~\ref{ass:interp}, $\Phi(\cdot\mid t)$ is the object
estimated by kime-phase tomography from the repeated-measurement
records $\{Y_{j,k}\}$; the identifiability, deconvolution, and
Cram\'er--Rao theory for this estimation problem is developed in
\cite{KPTms} and is taken as given here.

\subsection{The action--angle dictionary}

The following elementary lemma, which does not appear explicitly in
\cite{KPTms,WDWms,DiracKime}, although all of its ingredients do and it's important in this study.

\begin{lemma}[The kime cone is an action--angle
chart]\label{lem:actionangle}
Let $J=\tfrac12 t^2=\tfrac12|\kappa|^2$ and define
\[
  \Psi:\Sone\times(0,\infty)\longrightarrow\R^2\setminus\{0\},
  \qquad
  \Psi(\kimephase,J)=(q,p)
  =\bigl(\sqrt{2J}\,\sin\kimephase,\;\sqrt{2J}\,\cos\kimephase\bigr).
\]
Then $\Psi$ is a diffeomorphism and
\begin{equation}\label{eq:pullback}
  \Psi^{*}(\dd q\wedge\dd p)=\dd\kimephase\wedge\dd J,
  \qquad\text{hence}\qquad
  \{\kimephase,J\}=1
\end{equation}
with respect to the standard symplectic structure
$\omega_0=\dd q\wedge\dd p$ on the punctured plane. Moreover the
Liouville measure corresponds to the kime measure:
$\dd q\,\dd p=\dd J\,\dd\kimephase=t\,\dd t\,\dd\kimephase
=2\pi\,\dd\mu_{g_0}$.
\end{lemma}

\begin{proof}
With $t=\sqrt{2J}$, we compute
$\dd q=\sqrt{2J}\cos\kimephase\,\dd\kimephase
+(2J)^{-1/2}\sin\kimephase\,\dd J$ and
$\dd p=-\sqrt{2J}\sin\kimephase\,\dd\kimephase
+(2J)^{-1/2}\cos\kimephase\,\dd J$. Wedging,
\[
  \dd q\wedge\dd p
  =\cos^2\kimephase\,\dd\kimephase\wedge\dd J
   -\sin^2\kimephase\,\dd J\wedge\dd\kimephase
  =\dd\kimephase\wedge\dd J .
\]
Smooth invertibility on the stated domains is clear (polar
coordinates). The bracket statement follows because in any chart in
which $\omega_0$ takes the Darboux form $\dd x\wedge\dd y$ one has
$\{x,y\}=1$; here $(x,y)=(\kimephase,J)$. Finally
$\dd J=t\,\dd t$, and \eqref{eq:conemetric} carries the Haar factor
$1/2\pi$, giving the last display.
\end{proof}

\emph{Orientation convention.}\label{rem:orientation}
Two natural identifications of the punctured kime plane with a
one-DOF phase space differ by orientation. The map $\Psi$ above is
chosen so that $\{\kimephase,J\}=+1$, matching the sign convention of
the conjugate pair $\{\theta^{xy},S_z\}=1$ in Problem~(III)
\cite{CarcassiAidalaBook}; the alternative identification
$(q,p)=(\operatorname{Re}\kappa,\operatorname{Im}\kappa)
=(t\cos\kimephase,\,t\sin\kimephase)$ is holomorphic in $\kappa$ and
gives $\{J,\kimephase\}=+1$, i.e., the opposite orientation. The two
choices are exchanged by the reflection $\kappa\mapsto\bar\kappa$ and
carry the \emph{same} unsigned area measure
$|\dd\kimephase\wedge\dd J|=t\,\dd t\,\dd\kimephase$, so every
measure-theoretic and entropic statement below is independent of the
choice. The distinction matters only for the K\"ahler normal form of
Section~\ref{ssec:kahler}, where it is made explicit
(Proposition~\ref{prop:kahler}).

\emph{Scope of the dictionary.}\label{rem:dictionary}
Lemma~\ref{lem:actionangle} states that the kime chart
$(\kimephase,t)$ with the cone measure is \emph{exactly} the
action--angle chart of one classical DOF (for the harmonic oscillator,
$J$ is the action and $\kimephase$ the angle; for a general
one-DOF Hamiltonian with compact regular energy levels, the
Liouville--Arnold theorem \cite[Ch.~10]{Arnold} supplies an
action--angle chart with the same symplectic normal form). Under
Assumption~\ref{ass:interp} the identification of the \emph{latent
statistical} phase with the \emph{mechanical} angle is a modeling
step: it asserts that trial-to-trial variability of a repeated
experiment is variability of the angle variable at (approximately)
fixed action. This is the precise, falsifiable sense in which the
kime representation ``lives on'' the phase spaces of Problems
(I)--(III), and every theorem below separates cleanly into a
representation-level statement (unconditional) and this
identification (a postulate, in the same spirit as the ground-state
matching postulate of \cite{WDWms}).

\begin{definition}[Kime representation of a state]\label{def:kimestate}
A \emph{state} of one DOF is a probability density $\rho$ on
$(\R^2,\dd q\,\dd p)$. Its \emph{kime representation} is the density
$\tilde\rho=\rho\circ\Psi$ on $\Sone\times(0,\infty)$ with respect to
$\dd\kimephase\,\dd J$ (no Jacobian appears, by
Lemma~\ref{lem:actionangle}). The induced \emph{phase law at action
$J$} is the conditional density
\[
  \Phi(\kimephase\mid J)
  =\frac{\tilde\rho(\kimephase,J)}{\rho_J(J)},
  \qquad
  \rho_J(J)=\int_{-\pi}^{\pi}\tilde\rho(\kimephase,J)\,\dd\kimephase ,
\]
defined for $\rho_J(J)>0$. A state is \emph{phase-equipartitioned} if
$\Phi(\cdot\mid J)$ is the uniform (Haar) law for
$\rho_J$-a.e.~$J$, i.e., $\tilde\rho=\tilde\rho(J)$.
\end{definition}

\subsection{Circular Fisher information and phase diffusion}

\begin{definition}[Circular Fisher information]\label{def:fisher}
For a strictly positive \emph{(kime-phase) density}
$\Phi\in C^1(\Sone)$,
\[
  \Fish[\Phi]
  =\int_{-\pi}^{\pi}\frac{\bigl(\Phi'(\kimephase)\bigr)^2}
  {\Phi(\kimephase)}\,\dd\kimephase
  =\int_{-\pi}^{\pi}\Phi(\kimephase)
  \bigl(\partial_\kimephase\log\Phi(\kimephase)\bigr)^2\,
  \dd\kimephase .
\]
\end{definition}

\begin{lemma}[Amplitude/kinetic identity;
{\cite{WDWms}}]\label{lem:kinetic}
For strictly positive $\Phi\in C^1(\Sone)$,
\begin{equation}\label{eq:kinetic}
  \int_{-\pi}^{\pi}
  \bigl|\partial_\kimephase\sqrt{\Phi(\kimephase)}\bigr|^2
  \dd\kimephase=\tfrac14\,\Fish[\Phi],
  \qquad\text{hence}\qquad
  \Bigl\langle \tfrac{\hat p_\kimephase^{\,2}}{2\muI}
  \Bigr\rangle_{\!\sqrt{\Phi}}
  =\frac{\hbar^2}{8\muI}\,\Fish[\Phi],
\end{equation}
where $\hat p_\kimephase=-i\hbar\partial_\kimephase$ acts on the
periodic Sobolev space $H^1(\Sone)$ and $\muI>0$ is the phase inertia
of \cite{WDWms}, in which \eqref{eq:kinetic} appears as the
Fisher--kinetic identity underlying the potential-reconstruction
theorem.
\end{lemma}

\begin{proof}
$\partial_\kimephase\sqrt\Phi=\Phi'/(2\sqrt\Phi)$, so
$|\partial_\kimephase\sqrt\Phi|^2=(\Phi')^2/(4\Phi)$; integrate. The
second identity is the expectation of
$\hat p_\kimephase^2/(2\muI)$ in the real state $\sqrt\Phi$, using
integration by parts on $\Sone$ (boundary terms vanish by
periodicity).
\end{proof}

\begin{lemma}[Phase diffusion: moment decay, entropy production,
equipartition attractor]\label{lem:heat}
Let $\Phi_t$ solve the Fokker--Planck (heat) equation on $\Sone$
\cite{Risken},
$\partial_t\Phi_t=D\,\partial_\kimephase^2\Phi_t$, $D>0$, with
strictly positive $C^2$ initial datum $\Phi_0$. 
Then,
\begin{enumerate}[label=(\roman*)]
\item (Moment decay; the stochastic-phase suppression law of
\cite{DiracKime})
$\alpha_n(t)=e^{-Dn^2t}\alpha_n(0)$ for all $n\in\Z$.
\item (de Bruijn identity on the circle; cf.\
\cite{Stam1959,DemboCoverThomas})
$\dfrac{\dd}{\dd t}\,\ent[\Phi_t]=D\,\Fish[\Phi_t]\ \ge 0$, with
equality at time $t$ iff $\Phi_t$ is the uniform density
$\Phi_\infty\equiv 1/2\pi$.
\item (Exponential equipartition)
$\chi^2(\Phi_t\,\|\,\Phi_\infty)
:=2\pi\!\int_{-\pi}^{\pi}\!\bigl(\Phi_t-\tfrac1{2\pi}\bigr)^2
\dd\kimephase
\le e^{-2Dt}\,\chi^2(\Phi_0\,\|\,\Phi_\infty)$,
and $\ent[\Phi_t]\uparrow\log 2\pi=\max\ent$.
\end{enumerate}
\end{lemma}

\begin{proof}
(i) Multiply the equation by $e^{in\kimephase}$ and integrate by parts
twice on $\Sone$ (all boundary terms vanish by periodicity):
$\dot\alpha_n=-Dn^2\alpha_n$.

(ii) Positivity of $\Phi_t$ for $t>0$ is standard (strong maximum
principle for the heat semigroup on $\Sone$). Differentiating,
\[
  \frac{\dd}{\dd t}\ent[\Phi_t]
  =-\int(\partial_t\Phi_t)(1+\log\Phi_t)\,\dd\kimephase
  =-D\int \Phi_t''\,\log\Phi_t\,\dd\kimephase
  =D\int\frac{(\Phi_t')^2}{\Phi_t}\,\dd\kimephase
  =D\,\Fish[\Phi_t],
\]
where we used $\int\Phi_t''\,\dd\kimephase=0$ and one integration by
parts; all boundary terms vanish by periodicity. $\Fish[\Phi]=0$ iff
$\Phi'\equiv0$ iff $\Phi$ is uniform.

(iii) With trigonometric moments normalized as in
Definition~\ref{def:phaselaw}, Parseval's identity gives
$\chi^2(\Phi_t\|\Phi_\infty)=\sum_{n\ne0}|\alpha_n(t)|^2$, so
by~(i)
\[
  \chi^2(\Phi_t\|\Phi_\infty)
  =\sum_{n\ne0}e^{-2Dn^2t}|\alpha_n(0)|^2
  \le e^{-2Dt}\,\chi^2(\Phi_0\|\Phi_\infty).
\]
For the entropy limit, note first that
$\ent[\Phi]=\log2\pi-\KL(\Phi\,\|\,\Phi_\infty)$ by direct
computation, and second that
$\KL\le\log(1+\chi^2)\le\chi^2$ by Jensen's inequality applied to the
concave logarithm \cite{CoverThomas}. Hence
$\log2\pi\ \ge\ \ent[\Phi_t]\ \ge\ \log2\pi
-\chi^2(\Phi_t\|\Phi_\infty)\ \longrightarrow\ \log2\pi$,
and monotonicity of the convergence is (ii). Uniqueness of the
maximizer ($\ent[\Phi]\le\log2\pi$ with equality iff $\Phi$ uniform)
is Jensen's inequality applied to $-\log$.
\end{proof}

\emph{ Lemma \ref{lem:heat}
interpretation.}\label{int:foundations}
Under Assumption~\ref{ass:interp}, Lemma~\ref{lem:heat} is a statement
about \emph{experimental reproducibility}: if the latent trial phase
diffuses between repetitions (the stochastic-projection mechanism of
\cite{DiracKime}), then the empirical phase law loses structure at the
universal rate $e^{-Dn^2t}$ per harmonic, its entropy grows
monotonically, and the maximally reproducibility-agnostic (Haar)
law is the unique attractor. Lemma~\ref{lem:actionangle} will convert
this into the dynamical/thermodynamic half of the uncertainty problem
in Section~\ref{ssec:thermo}.

\section{Problem I (uncertainty): the classical uncertainty principle
in the kime representation}\label{sec:uncertainty}

We first recall the flat-space entropic uncertainty principle in the
form used by \cite{CarcassiAidalaBook,CarcassiAidalaEntropy}, then
prove its exact analogue on the kime cylinder $\Sone\times\R$
(Theorem~\ref{thm:cylinder}), a sharp circular Fisher inequality
(Theorem~\ref{thm:sharpcircular}), and an exact non-canonical
uncertainty relation (Theorem~\ref{thm:noncanonical}). We then treat
multiple DOF (Theorems~\ref{thm:williamson}--\ref{thm:aggregate})
and the dynamical/thermodynamic conjecture
(Theorem~\ref{thm:mixing}), and state the remaining open problems
(Problems~\ref{prob:expected-bracket}--\ref{prob:capacity}).

\subsection{Maximum-entropy lemmas and the flat benchmark}

\begin{lemma}[Gaussian maximum entropy; \cite{CoverThomas}]
\label{lem:gaussmax}
Let $\rho$ be a probability density on $\R$ (w.r.t.\ Lebesgue) with
finite variance $\sigma^2>0$ and $\ent[\rho]>-\infty$. Then
$\ent[\rho]\le\tfrac12\log(2\pi e\sigma^2)$, with equality iff $\rho$
is Gaussian with variance $\sigma^2$.
\end{lemma}

\begin{proof}
Let $g$ be the Gaussian density with the same mean $m$ and variance
$\sigma^2$. Then
$0\le \KL(\rho\|g)
=-\ent[\rho]-\int\rho\log g$, and
$-\int\rho\log g
=\tfrac12\log(2\pi\sigma^2)
+\E_\rho\!\bigl[(x-m)^2\bigr]/(2\sigma^2)
=\tfrac12\log(2\pi\sigma^2)+\tfrac12$. Rearranging gives the bound;
equality in $\KL\ge0$ holds iff $\rho=g$ a.e.
\end{proof}

\begin{lemma}[Entropy subadditivity; \cite{CoverThomas}]
\label{lem:subadd}
Let $\rho$ be a probability density on a product measure space
$(X\times Y,\Leb_X\otimes\Leb_Y)$ with marginals $\rho_X,\rho_Y$ and
$\ent[\rho]>-\infty$. Then
$\ent[\rho]\le\ent[\rho_X]+\ent[\rho_Y]$, with equality iff
$\rho=\rho_X\otimes\rho_Y$ a.e.
\end{lemma}

\begin{proof}
$\KL(\rho\,\|\,\rho_X\otimes\rho_Y)
=\int\rho\log\rho-\int\rho\log(\rho_X\rho_Y)
=-\ent[\rho]+\ent[\rho_X]+\ent[\rho_Y]\ \ge 0$,
using that the $X$- and $Y$-marginal integrals of $\rho\log\rho_X$
and $\rho\log\rho_Y$ reduce to marginal entropies by Fubini.
Equality iff $\rho=\rho_X\otimes\rho_Y$ a.e.
\end{proof}

\begin{theorem}[Flat entropic uncertainty principle; cf.\
\cite{CarcassiAidalaBook,CarcassiAidalaEntropy}]\label{thm:flat}
Let $\rho$ be a probability density on $(\R^2,\dd q\,\dd p)$ with
finite marginal variances $\sigma_q^2,\sigma_p^2$ and
$\ent[\rho]>-\infty$. Then
\begin{equation}\label{eq:flatUP}
  \sigma_q\,\sigma_p\;\ge\;\frac{1}{2\pi e}\,e^{\ent[\rho]},
\end{equation}
with equality iff $\rho$ is a Gaussian product density. Since
Hamiltonian flows preserve $\dd q\,\dd p$, the right-hand side is a
constant of motion, while the factors on the left individually evolve.
\end{theorem}

\begin{proof}
Lemma~\ref{lem:subadd} and then Lemma~\ref{lem:gaussmax} on each
marginal give
$\ent[\rho]\le\tfrac12\log(2\pi e\sigma_q^2)
+\tfrac12\log(2\pi e\sigma_p^2)=\log(2\pi e\,\sigma_q\sigma_p)$.
Then, exponentiate and notice that equality in \eqref{eq:flatUP} requires equality in both lemmas
\ref{lem:gaussmax}
and \ref{lem:subadd}.
Invariance
of $\ent$ under measure-preserving flows is the change-of-variables
identity $\ent[\rho\circ\phi^{-1}]=\ent[\rho]$ for
$\phi$ preserving $\dd q\,\dd p$ (Liouville's theorem).
\end{proof}

\emph{Theorem \ref{thm:flat}
interpretation.}
Inequality \eqref{eq:flatUP} is a purely classical, Liouville-based
statement; it should be distinguished from the Hirschman--Beckner
entropic uncertainty principle for Fourier-conjugate quantum
observables, for which see \cite{Folland}. The two share the
maximum-entropy mechanism but not the underlying invariance.

\subsection{The circular sector: sharp inequalities on
$\Sone$ and on the kime cylinder}

The kime phase is compact, and the pair
$(\kimephase,p_\kimephase)$, an angle and its conjugate momentum, is the simplest pair for which the flat formulation of
Theorem~\ref{thm:flat} is \emph{not} directly meaningful ($\kimephase$
is multivalued; its ``variance'' is convention-dependent). This is
exactly the non-canonical obstruction isolated in Problem~(I)(a). The
circular-statistics objects of Definition~\ref{def:phaselaw} resolves
it.

\begin{lemma}[von Mises maximum entropy; cf.\ \cite{MardiaJupp}]
\label{lem:vmmax}
Let $\Phi$ be a probability density on $(\Sone,\dd\kimephase)$ with
mean resultant length $r\in[0,1)$ and (if $r>0$) mean direction
$\meandir$, and $\ent[\Phi]>-\infty$. Let
$A(\kappa)=I_1(\kappa)/I_0(\kappa)$ denote the Bessel ratio, a strictly
increasing bijection $A:[0,\infty)\to[0,1)$ \cite[Sec.~3.5]{MardiaJupp},
and set $\kappa(r)=A^{-1}(r)$. Then,
\begin{equation}\label{eq:hc}
  \ent[\Phi]\;\le\; h_c(r)
  \;:=\;\log\!\bigl(2\pi I_0(\kappa(r))\bigr)-\kappa(r)\,r,
\end{equation}
with equality iff
$\Phi=\vM(\cdot;\meandir,\kappa(r))
=\exp\{\kappa(r)\cos(\kimephase-\meandir)\}/\bigl(2\pi
I_0(\kappa(r))\bigr)$. The function $h_c$ is strictly decreasing on
$(0,1)$ with $h_c'(r)=-\kappa(r)$, $h_c(0)=\log2\pi$, and
$h_c(r)\to-\infty$ as $r\uparrow1$.
\end{lemma}

\begin{proof}
Let $g=\vM(\cdot;\meandir,\kappa(r))$ (for $r=0$, $g$ is uniform and
the bound is Jensen's inequality). Then
$0\le\KL(\Phi\|g)=-\ent[\Phi]-\int\Phi\log g$ and
\[
  -\int\Phi\log g
  =\log\!\bigl(2\pi I_0(\kappa(r))\bigr)
  -\kappa(r)\,\E_\Phi\bigl[\cos(\Theta-\meandir)\bigr]
  =\log\!\bigl(2\pi I_0(\kappa(r))\bigr)-\kappa(r)\,r,
\]
because $\E_\Phi\cos(\Theta-\meandir)=r$ by definition of
$(r,\meandir)$. Equality iff $\Phi=g$. For the derivative, write
$h_c(r)=\log(2\pi I_0(\kappa))-\kappa r$ with $A(\kappa)=r$; then,
using $I_0'=I_1$,
$h_c'(r)=A(\kappa)\kappa'(r)-r\kappa'(r)-\kappa(r)=-\kappa(r)<0$ for
$r>0$. The boundary values are immediate from $\kappa(0)=0$ and
$\kappa(r)\to\infty$.
\end{proof}

\begin{definition}[Circular entropy width]\label{def:width}
For $r\in[0,1)$ define $\Lambda(r):=e^{h_c(r)}\in(0,2\pi]$. By
Lemma~\ref{lem:vmmax}, $\Lambda$ is strictly decreasing,
$\Lambda(0)=2\pi$, and $\Lambda(r)\downarrow0$ as $r\uparrow1$.
That is, 
$\Lambda(r)$ is the effective support length of the most disordered
circular law compatible with concentration $r$.
\end{definition}

\begin{proposition}[Flat limit of the width]\label{prop:flatlimit}
As $r\uparrow1$ (equivalently $\kappa=\kappa(r)\to\infty$),
\[
  \Lambda(r)=\sqrt{\frac{2\pi e}{\kappa(r)}}\,\bigl(1+O(\kappa^{-1})\bigr),
\]
and for the extremal law $\vM(\cdot;\meandir,\kappa)$ one has
$\Var(\Theta-\meandir)=\kappa^{-1}\bigl(1+O(\kappa^{-1})\bigr)$, so
that $\Lambda(r)\sim\sqrt{2\pi e}\;\sigma_\Theta$ and the cylinder
theorem below degenerates to the flat bound \eqref{eq:flatUP}.
\end{proposition}

\begin{proof}
Inserting the standard asymptotics
$I_0(\kappa)=e^{\kappa}(2\pi\kappa)^{-1/2}(1+O(\kappa^{-1}))$ and
$A(\kappa)=1-\tfrac1{2\kappa}+O(\kappa^{-2})$
\cite[9.7.1]{AbramowitzStegun} into \eqref{eq:hc}
yields
\[
  h_c=\log2\pi+\kappa-\tfrac12\log(2\pi\kappa)-\kappa A(\kappa)
  +O(\kappa^{-1})
  =\tfrac12\log\!\frac{2\pi e}{\kappa}+O(\kappa^{-1}).
\]
Exponentiating gives the first claim. Observe that the variance asymptotics follow
from the Laplace approximation of the von Mises law around
$\meandir$ (Gaussian with variance $1/\kappa$), justified by the same
Bessel asymptotics. Matching with Theorem~\ref{thm:flat} is then the
computation
$\Lambda(r)\sigma_p\sim\sqrt{2\pi e}\,\sigma_\Theta\sigma_p$.
\end{proof}

\begin{theorem}[Kime-cylinder entropic uncertainty
principle]\label{thm:cylinder}
Let $\rho$ be a probability density on
$\bigl(\Sone\times\R,\;\dd\kimephase\,\dd p\bigr)$ with
$\ent[\rho]>-\infty$, angular marginal $\Phi$ with mean resultant
length $r\in[0,1)$, and momentum marginal with finite variance
$\sigma_p^2>0$. Then,
\begin{equation}\label{eq:cylUP}
  \Lambda(r)\cdot\sigma_p\;\ge\;\frac{e^{\ent[\rho]}}{\sqrt{2\pi e}},
\end{equation}
with equality iff
$\rho=\vM(\cdot;\meandir,\kappa(r))\otimes
\mathcal N(m,\sigma_p^2)$ for some $m\in\R$.
Moreover, if $\rho$ evolves by any Hamiltonian flow on the cylinder
$T^*\Sone$ (symplectic form $\dd\kimephase\wedge\dd p$), the
right-hand side of \eqref{eq:cylUP} is a constant of motion.
\end{theorem}

\begin{proof}
By Lemma~\ref{lem:subadd},
$\ent[\rho]\le\ent[\Phi]+\ent[\rho_p]$. Bounding the two summands by
Lemmas~\ref{lem:vmmax} and~\ref{lem:gaussmax},
$\ent[\rho]\le h_c(r)+\tfrac12\log(2\pi e\sigma_p^2)
=\log\bigl(\Lambda(r)\,\sigma_p\sqrt{2\pi e}\bigr)$. Again, we exponentiate the terms.
Equality forces equality in all three inequalities, i.e., independence
with extremal marginals. Any Hamiltonian flow preserves the Liouville
measure $\dd\kimephase\,\dd p$, hence preserves $\ent[\rho]$ by the
change-of-variables identity.
\end{proof}

\emph{Compactness is a feature, not a artifact.}\label{rem:compact}
Since $\Lambda(r)\le 2\pi$, inequality \eqref{eq:cylUP} contains an
absolute momentum floor
$\sigma_p\ge e^{\ent[\rho]}/(2\pi\sqrt{2\pi e})$. On a compact angle,
total delocalization of the phase cannot absorb an unbounded share of
the invariant entropy. This is the precise structural difference from
the flat case that Problem (I)(a) anticipated for non-canonical
variables. Compact topology converts the uncertainty \emph{trade-off}
into a trade-off with a hard floor. Under
Assumption~\ref{ass:interp}, $r$ is the concentration of the empirical
phase law of the repeated experiment and is directly estimable
\cite{KPTms}.
Therefore, the kime entropic uncertainty principle  \eqref{eq:cylUP} is a \emph{testable inequality}.

\begin{theorem}[Sharp circular Fisher (Cram\'er--Rao--type)
uncertainty relation]\label{thm:sharpcircular}
Let $\Phi\in C^1(\Sone)$ be strictly positive with mean resultant
length $r>0$ and mean direction $\meandir$. Then
\begin{equation}\label{eq:sharpcircular}
  \Fish[\Phi]\cdot
  \E_\Phi\!\bigl[\sin^2(\Theta-\meandir)\bigr]\;\ge\;r^2 ,
\end{equation}
with equality iff $\Phi=\vM(\cdot;\meandir,\kappa)$ for some
$\kappa>0$. Equivalently, in terms of trigonometric moments,
$\Fish[\Phi]\ge 2r^2/\bigl(1-\operatorname{Re}\,
e^{-2i\meandir}\alpha_2\bigr)$.
\end{theorem}

\begin{proof}
Integrating by parts on $\Sone$ (periodic boundary terms vanish),
\[
  \int_{-\pi}^{\pi}\Phi'(\kimephase)\,\sin(\kimephase-\meandir)\,
  \dd\kimephase
  =-\int_{-\pi}^{\pi}\Phi(\kimephase)\cos(\kimephase-\meandir)\,
  \dd\kimephase=-r .
\]
By Cauchy--Schwarz,
\[
  r^2=\Bigl(\int\frac{\Phi'}{\sqrt\Phi}\cdot
  \sqrt\Phi\,\sin(\kimephase-\meandir)\,\dd\kimephase\Bigr)^{\!2}
  \le\int\frac{(\Phi')^2}{\Phi}\,\dd\kimephase\;
  \int\Phi\sin^2(\kimephase-\meandir)\,\dd\kimephase ,
\]
which is \eqref{eq:sharpcircular}. Equality holds iff
$\Phi'/\Phi=c\,\sin(\kimephase-\meandir)$ a.e.\ for some constant $c$,
i.e., $\log\Phi=-c\cos(\kimephase-\meandir)+\mathrm{const}$, a von
Mises law. Consistency with $\E\cos(\Theta-\meandir)=r>0$ forces
$c<0$, i.e., concentration at $\meandir$ with $\kappa=-c>0$. Note that for
$\Phi=\vM(\cdot;\meandir,\kappa)$ we can directly check $\Fish=\kappa^2\,\E\sin^2(\Theta-\meandir)$ and
$\E\sin^2(\Theta-\meandir)=A(\kappa)/\kappa$. Hence,
$\Fish\cdot\E\sin^2=A(\kappa)^2=r^2$, see the von Mises Fisher
computation in \cite{WDWms}. The corresponding trigonometric moment form follows from
$\sin^2 x=\tfrac12(1-\cos2x)$.
\end{proof}

\begin{corollary}[Kime-native kinetic uncertainty
bound]\label{cor:kineticbound}
Under the ground-state matching postulate of \cite{WDWms}
($\varphi_0=\sqrt{\Phi}$), the kinetic energy of the reconstructed
phase state obeys
\[
  \Bigl\langle\frac{\hat p_\kimephase^{\,2}}{2\muI}
  \Bigr\rangle_{\!\sqrt\Phi}
  \;=\;\frac{\hbar^2}{8\muI}\,\Fish[\Phi]
  \;\ge\;\frac{\hbar^2}{8\muI}\,
  \frac{r^2}{\E_\Phi[\sin^2(\Theta-\meandir)]}
  \;\ge\;\frac{\hbar^2 r^2}{8\muI},
\]
with the first inequality saturated exactly by von Mises phase laws.
Concentration of the empirical phase law of a repeated experiment
therefore imposes a quantitative lower bound on the kinetic part of
the reconstructed phase Hamiltonian of \cite{WDWms}.
\end{corollary}

\begin{proof}
Combine Lemma~\ref{lem:kinetic} with
Theorem~\ref{thm:sharpcircular} and $\E\sin^2\le1$.
\end{proof}

\subsection{Non-canonical variables: the geometric-mean bracket
theorem}\label{ssec:noncanonical}

Problem~(I)(a) conjectures a minimum-uncertainty relation for a
non-canonical pair $(u,v)$ ``based on the \emph{Poisson bracket} (more precisely,
the \emph{expectation of the Poisson bracket}).'' Theorem \ref{thm:noncanonical}
resolves the diffeomorphic case exactly and shows that the correct
universal correction is the \emph{geometric mean}
$\exp\E_\rho\log|\{u,v\}|$ of the bracket, not its arithmetic
expectation. The two coincide precisely in the linear (in particular
\emph{canonical}) case.

\begin{lemma}[Entropy under phase-space
reparametrization]\label{lem:jacobian}
Let $U\subseteq\R^2$ be open, $\rho$ a probability density supported
in $U$ (w.r.t.\ $\dd q\,\dd p$) with $\ent[\rho]$ finite, and
$(u,v):U\to\R^2$ a $C^1$ diffeomorphism onto its image with Jacobian
determinant
$\det\frac{\partial(u,v)}{\partial(q,p)}=\{u,v\}\ne0$ on $U$. Let
$\rho^{(u,v)}$ denote the pushforward density w.r.t.\ $\dd u\,\dd v$.
If $\E_\rho\bigl|\log|\{u,v\}|\bigr|<\infty$, then
\begin{equation}\label{eq:entropyshift}
  \ent\bigl[\rho^{(u,v)}\bigr]
  =\ent[\rho]+\E_\rho\log\bigl|\{u,v\}\bigr| .
\end{equation}
\end{lemma}

\begin{proof}
Write $\Psi=(u,v)$ and $J_\Psi=\{u,v\}$. The pushforward density is
$\rho^{(u,v)}=\bigl(\rho/|J_\Psi|\bigr)\circ\Psi^{-1}$. By the change
of variables formula,
\[
  \ent[\rho^{(u,v)}]
  =-\int_{\Psi(U)}\frac{\rho}{|J_\Psi|}
  \log\frac{\rho}{|J_\Psi|}\Bigm|_{\Psi^{-1}}\dd u\,\dd v
  =-\int_U\rho\,\bigl(\log\rho-\log|J_\Psi|\bigr)\,\dd q\,\dd p ,
\]
which is \eqref{eq:entropyshift}; the integrability hypothesis
justifies splitting the integral.
\end{proof}

\begin{theorem}[Non-canonical uncertainty relation with geometric-mean
bracket]\label{thm:noncanonical}
In the setting of Lemma~\ref{lem:jacobian}, suppose additionally that
the pushforward marginals of $u$ and $v$ have finite variances
$\sigma_u^2,\sigma_v^2$. Then
\begin{equation}\label{eq:noncanUP}
  \sigma_u\,\sigma_v
  \;\ge\;
  \frac{1}{2\pi e}\;e^{\ent[\rho]}\;
  \exp\Bigl(\E_\rho\log\bigl|\{u,v\}\bigr|\Bigr),
\end{equation}
with equality iff the pushforward of $\rho$ under $(u,v)$ is a
Gaussian product density. In particular:
\begin{enumerate}[label=(\alph*)]
\item if $(u,v)$ is canonical ($\{u,v\}\equiv1$), \eqref{eq:noncanUP}
reduces to \eqref{eq:flatUP};
\item if $(u,v)$ is linear, $\{u,v\}$ is constant and the correction
equals $|\{u,v\}|=\E_\rho|\{u,v\}|$;
\item in general, by Jensen's inequality
$\exp\E_\rho\log|\{u,v\}|\le\E_\rho|\{u,v\}|$, so the conjectured
bound with the arithmetic expectation of the bracket is a
\emph{strictly stronger} statement than \eqref{eq:noncanUP} and does
not follow from entropy methods alone.
\end{enumerate}
\end{theorem}

\begin{proof}
Apply Theorem~\ref{thm:flat} to the pushforward density (a legitimate
density on $\R^2$) and substitute \eqref{eq:entropyshift}. Items
(a)--(b) are immediate; (c) is Jensen applied to the concave
logarithm.
\end{proof}

\begin{problem}[Status of the expected-bracket
form]\label{prob:expected-bracket}
Determine the largest class of pairs $(u,v)$ and states $\rho$ for
which the strengthened inequality
$\sigma_u\sigma_v\ge(2\pi e)^{-1}e^{\ent[\rho]}\,
\E_\rho|\{u,v\}|$ holds, and exhibit either a proof for a natural
class beyond the linear one or an explicit counterexample. In the
kime representation the natural test family is
$(u,v)=(\text{a circular function of }\kimephase,\;J)$, for which the
compactness corrections are controlled by
Theorem~\ref{thm:cylinder}.
In particular, formulate and prove the
correct statement when $(u,v)$ is \emph{not} injective (winding
angle), the conjectured mechanism is a holonomy correction
quantized in units of the circulation $\oint\dd\kimephase=2\pi$,
i.e., an additive term $\log(2\pi w)$ for winding number $w$, whose
precise form should follow by applying Lemma~\ref{lem:jacobian} on a
fundamental domain and Lemma~\ref{lem:vmmax} on the quotient.
\end{problem}

\subsection{Multiple degrees of freedom: Williamson form, Fischer
inequality, and the symplectic Schur--Horn problem}
\label{ssec:multidof}

Fix $n\ge1$ and coordinates
$z=(q^1,p_1,\dots,q^n,p_n)\in\R^{2n}$ ordered by DOF, so that a
covariance matrix $\Sigma\in\R^{2n\times2n}$, $\Sigma\succ0$,
decomposes into $2\times2$ blocks $\Sigma_{jk}$,
$j,k=1,\dots,n$, with $\Sigma_{jj}$ the \emph{within-DOF} block of
DOF $j$ and $\Sigma_{jk}$ ($j\ne k$) the cross-DOF correlation
blocks, the decomposition singled out in Problem~(I)(b),
Eq.~(1.229) of \cite{CarcassiAidalaBook}. In this ordering the
symplectic form is
$\Omega_n=\bigoplus_{j=1}^n
\begin{psmallmatrix}0&1\\-1&0\end{psmallmatrix}$ and
$\Sp(2n,\R)=\{S:\,S\Omega_nS^{\!\top}=\Omega_n\}$.

\begin{theorem}[Williamson normal form; \cite{Williamson1936}]
\label{thm:williamson}
For every $\Sigma\succ0$ there exists $S\in\Sp(2n,\R)$ and unique
$\nu_1\ge\cdots\ge\nu_n>0$ (the \emph{symplectic eigenvalues},
the positive spectrum of $i\Omega_n^{-1}\Sigma$ up to sign) with
$S\Sigma S^{\!\top}
=\operatorname{diag}(\nu_1,\nu_1,\dots,\nu_n,\nu_n)$.
\end{theorem}

\begin{lemma}[Gaussian entropy and symplectic
invariants]\label{lem:gaussentropyformula}
If $\rho$ is the Gaussian density on $\R^{2n}$ with covariance
$\Sigma$, then
$\ent[\rho]=\tfrac12\log\bigl((2\pi e)^{2n}\det\Sigma\bigr)$ and
$\det\Sigma=\prod_{j=1}^n\nu_j^2$. Both $\ent[\rho]$ and the
multiset $\{\nu_j\}$ are invariant under every linear Hamiltonian
evolution $\Sigma\mapsto S\Sigma S^{\!\top}$, $S\in\Sp(2n,\R)$.
\end{lemma}

\begin{proof}
The entropy formula is the standard Gaussian computation (diagonalize
$\Sigma$ orthogonally and apply Lemma~\ref{lem:gaussmax}
coordinatewise, plus subadditivity with equality for independent
coordinates). Since $\det S=1$ for $S\in\Sp(2n,\R)$,
$\det\Sigma=\det(S\Sigma S^{\!\top})=\prod_j\nu_j^2$ by
Theorem~\ref{thm:williamson}. Invariance of the symplectic spectrum:
if $S_0\in\Sp$, then
$\Omega_n^{-1}(S_0\Sigma S_0^{\!\top})
=S_0^{-\!\top}(\Omega_n^{-1}\Sigma)S_0^{\!\top}$ using
$S_0^{\!\top}\Omega_n S_0=\Omega_n$, a similarity transformation, so
the spectrum of $i\Omega_n^{-1}\Sigma$ is unchanged.
\end{proof}

\begin{lemma}[Fischer inequality; see \cite{HornJohnson}]
\label{lem:fischer}
For $\Sigma\succ0$ partitioned into diagonal blocks
$\Sigma_{11},\dots,\Sigma_{nn}$ (any sizes),
$\det\Sigma\le\prod_{j=1}^n\det\Sigma_{jj}$, with equality iff
$\Sigma$ is block diagonal.
\end{lemma}

\begin{proof}
It suffices to treat $n=2$ and induct. Writing the Schur complement
$\Sigma/\Sigma_{11}=\Sigma_{22}-\Sigma_{21}\Sigma_{11}^{-1}\Sigma_{12}$,
one has $\det\Sigma=\det\Sigma_{11}\det(\Sigma/\Sigma_{11})$ and
$0\prec\Sigma/\Sigma_{11}\preceq\Sigma_{22}$. Monotonicity of the
determinant on the positive-semidefinite order
\cite[Cor. 7.7.4]{HornJohnson} gives
$\det(\Sigma/\Sigma_{11})\le\det\Sigma_{22}$, with equality iff
$\Sigma_{21}\Sigma_{11}^{-1}\Sigma_{12}=0$, i.e., $\Sigma_{12}=0$.
\end{proof}

\begin{theorem}[Aggregate multi-DOF uncertainty
floor]\label{thm:aggregate}
Let $\rho$ be a probability density on $\R^{2n}$ (with respect to
$\dd^{2n}z$) with finite covariance $\Sigma\succ0$ and
$\ent[\rho]>-\infty$, and let
$\nu_1\ge\dots\ge\nu_n>0$ be the symplectic eigenvalues
of~$\Sigma$.
Define the \emph{within-DOF uncertainty} of DOF~$j$ as
\[
  u_j:=\sqrt{\det\Sigma_{jj}}
  =\sqrt{\sigma^2_{q^j}\sigma^2_{p_j}-\Cov(q^j,p_j)^2},
\]
the area scale of the $j$th marginal covariance ellipse. Then,
\begin{equation}\label{eq:aggregate}
  \prod_{j=1}^n u_j
  \;\ge\;\sqrt{\det\Sigma}
  \;=\;\prod_{j=1}^n\nu_j
  \;\ge\;\frac{e^{\ent[\rho]}}{(2\pi e)^{n}},
\end{equation}
where the first inequality is an equality iff there are no
cross-DOF correlations ($\Sigma_{jk}=0$ for $j\ne k$); the middle
quantity is invariant under all linear Hamiltonian flows; and the
last inequality is an equality iff $\rho$ is Gaussian. Consequently,
along any linear Hamiltonian evolution of a state that is initially
Gaussian, uncorrelated across DOF, and equipartitioned
($\nu_j=\nu$ for all $j$, $u_j(0)=\nu$),
\[
  \prod_{j=1}^n u_j(t)\;\ge\;\nu^{\,n}=\prod_j u_j(0)
  \qquad\text{for all } t,
\]
with equality at time $t$ iff the state is again uncorrelated across
DOF at time $t$. \emph{The product of within-DOF uncertainties can
only be raised above its initial equipartition value, and only by the
creation of cross-DOF correlations.} For $n=1$ the product is a
single factor and $u_1(t)=\nu_1$. Identically, within one DOF, linear
Hamiltonian flow conserves the uncertainty area exactly.
\end{theorem}

\begin{proof}
The first inequality is Lemma~\ref{lem:fischer} applied to the
$2\times2$ block partition, together with
$\det\Sigma=\prod_j\nu_j^2$ (Lemma~\ref{lem:gaussentropyformula}).
The last inequality is the maximum-entropy bound
$\ent[\rho]\le\tfrac12\log((2\pi e)^{2n}\det\Sigma)$, proved exactly
as in Lemma~\ref{lem:gaussmax} with the matching Gaussian of
covariance $\Sigma$ (equality iff $\rho$ Gaussian). Invariance of
$\prod\nu_j$ under $\Sigma\mapsto S\Sigma S^{\!\top}$ is
Lemma~\ref{lem:gaussentropyformula}. For the dynamical statement,
$\Sigma(t)=S_t\Sigma(0)S_t^{\!\top}$ with $S_t\in\Sp(2n,\R)$
preserves the symplectic spectrum $\{\nu\}$, so
$\prod_j u_j(t)\ge\prod_j\nu_j=\nu^n$ by the first inequality, with
the stated equality case. For $n=1$, $\Sp(2,\R)=\SL(2,\R)$, so
$u_1(t)^2=\det\Sigma(t)=\det\Sigma(0)=\nu_1^2$.
\end{proof}

\emph{Relevance to 
Problem (I)(b). }\label{rem:aggregatemeaning}
Theorem~\ref{thm:aggregate} proves the \emph{aggregate} (product)
version of both conjectures in Problem~(I)(b). The equi-partitioned
uncorrelated state minimizes the total within-DOF uncertainty over
its entire symplectic orbit, and any excess is exactly accounted for
by cross-DOF correlation (Fischer defect). What it does \emph{not}
decide is the \emph{per-DOF} refinement, whether an individual
$u_j(t)$ can dip below $\nu$ while others rise. That is the genuinely
open-problem content discussed in Problems \ref{prob:schurhorn} -- \ref{prob:capacity}.

\begin{problem}[Symplectic Schur--Horn problem for within-DOF
uncertainties]\label{prob:schurhorn}
Characterize, for fixed symplectic spectrum
$\nu_1\ge\cdots\ge\nu_n>0$, the attainable set
\[
  \mathcal U(\nu)
  =\Bigl\{\bigl(u_1(\Sigma),\dots,u_n(\Sigma)\bigr):
  \Sigma\in\text{the }\Sp(2n,\R)\text{-orbit with spectrum }\nu
  \Bigr\}\subset\R_{>0}^n ,
\]
where $u_j(\Sigma)=\sqrt{\det\Sigma_{jj}}$. In particular, we need to explore
(a) is $\min_j u_j\ge\nu_n$ on the whole orbit (so that no DOF can be
squeezed below the \emph{smallest} symplectic eigenvalue)?
(b) in the equipartitioned case $\nu_j\equiv\nu$, is
$u_j\ge\nu$ for \emph{every} $j$ (the per-DOF form of the
conjecture in \cite{CarcassiAidalaBook})?
And (c) describe $\mathcal U(\nu)$ by majorization-type inequalities,
in analogy with the Schur--Horn theorem, using the symplectic
eigenvalue technology of \cite{BhatiaJain2015}.

Some partial cues include (i) $\prod_j u_j\ge\prod_j\nu_j$
(Theorem~\ref{thm:aggregate}); (ii) $u_j\ge\nu_n$ would follow from
the interlacing-type bound ``every $2\times2$ symplectic compression
of $\Sigma$ has symplectic eigenvalue $\ge\nu_n$,'' a statement of
exactly the kind studied in \cite{BhatiaJain2015};  and (iii) for $n=1$,
$\mathcal U(\nu)=\{\nu\}$ (Theorem~\ref{thm:aggregate}).
\end{problem}

\begin{problem}[Kime/torus reformulation and estimability]
\label{prob:torus}
Via Lemma~\ref{lem:actionangle} applied per DOF, a state on
$\R^{2n}$ (off the coordinate axes) is a density on
$\mathbb T^n\times(0,\infty)^n$ in variables
$(\kimephase_1,\dots,\kimephase_n,J_1,\dots,J_n)$, and cross-DOF
correlations at fixed actions are encoded in the joint phase law on
$\mathbb T^n$, e.g., in the \emph{relative-phase moment matrix}
$R_{jk}=\E\bigl[e^{i(\Theta_j-\Theta_k)}\bigr]$, which is Hermitian, positive
semidefinite, unit diagonal. The solution may require progress in three directions. (a)~Express the quantities
$u_j,\nu_j$ of Problem~\ref{prob:schurhorn}, for Gaussian states,
in terms of $(R,\E J_1,\dots,\E J_n)$ in the small-dispersion regime,
and determine which functions of $\mathcal U(\nu)$ are identifiable
from repeated-measurement data under the multivariate extension of
the KPT observation model and its convolution identity
$M=F\Phi$ of \cite{KPTms} (circular convolution now acting on
$\mathbb T^n$); (b)~prove the multivariate
anchored-identifiability theorem (gauge group: independent rigid
rotations of each $\kimephase_j$, i.e., the torus
$\mathbb T^n$ acting diagonally on $R$ by conjugation with unimodular
diagonals, note $R$ itself is gauge-covariant while
$\operatorname{spec}R$ and $|R_{jk}|$ are gauge-invariant);
(c)~derive the Cram\'er--Rao bound for estimating
$(|R_{jk}|)_{j<k}$, extending the parametric Cram\'er--Rao theorem
of \cite{KPTms}, so that Problem~\ref{prob:schurhorn}(b)
acquires a statistically testable surrogate.
\end{problem}

\emph{Relation to symplectic capacities.}\label{rem:capacity}
For the covariance ellipsoid
$E_\Sigma=\{z:z^{\!\top}\Sigma^{-1}z\le1\}$, the linear symplectic
capacity equals $\pi\nu_n$ (the smallest symplectic eigenvalue sets
the Gromov width of the ellipsoid) \cite{Gromov1985,deGosson2009}.
Thus, Problem \ref{prob:schurhorn}(a) asks whether the within-DOF
uncertainty of every DOF dominates the capacity scale of the total
state. Under \emph{nonlinear} Hamiltonian flows $\Sigma$ is no longer
transported by $\Sp(2n,\R)$, but Gromov non-squeezing still bounds
the projection of the evolved support onto each conjugate plane.

\begin{problem}[Nonlinear invariant interpolating entropy and
capacity]\label{prob:capacity}
Construct a functional $\mathcal C[\rho]$ of states on $\R^{2n}$ such
that: (i)~$\mathcal C$ is invariant under all (possibly nonlinear)
Hamiltonian flows; (ii)~$\mathcal C$ reduces to $\pi\nu_n$ on
Gaussian states; (iii)~$\mathcal C$ lower-bounds
$2\pi\min_j u_j$ up to a universal constant; and (iv)~$\mathcal C$ is
expressible through the kime-torus data of
Problem~\ref{prob:torus} (hence estimable). Consider the following candidate functional, a
sublevel-set capacity of $\rho$ at the entropy-calibrated level
$e^{-\ent[\rho]}$, i.e.,
$\mathcal C[\rho]=c\bigl(\{\rho\ge e^{-\ent[\rho]}\}\bigr)$ for a
normalized symplectic capacity $c$.
Properties (i) and (ii) then
hold by symplectomorphism-invariance of $c$ and a direct Gaussian
computation, while (iii)--(iv) are open.
\end{problem}

\subsection{Dynamics and thermodynamics: equipartition as the
kime-diffusive attractor}\label{ssec:thermo}

The final component of Problem~(I) is the conjecture that
``Hamiltonian dynamics preserves, at least locally, the entropic
relationships that one would find at equilibrium,'' with
``equipartition of entropy over uncorrelated DOF'' as the conjectured
floor. The kime representation supplies both an exactly solvable
relaxation model (Lemma~\ref{lem:heat}) and a clean statement of what
Hamiltonian flow does and does not preserve.

\begin{definition}[Kime-deformed evolution]\label{def:kimedeformed}
Let $H$ be a one-DOF Hamiltonian admitting a global action--angle
chart $(\kimephase,J)\in\Sone\times\mathcal J$,
$\mathcal J\subseteq(0,\infty)$ open, with frequency
$\omega(J)=H'(J)$ (Liouville--Arnold, \cite[Ch.~10]{Arnold}). For
$\varepsilon\ge0$, the \emph{kime-deformed evolution} of a density
$\tilde\rho$ on $\Sone\times\mathcal J$ is
\begin{equation}\label{eq:kimedeformed}
  \partial_t\tilde\rho
  =-\,\omega(J)\,\partial_\kimephase\tilde\rho
  +\varepsilon\,\partial^2_\kimephase\tilde\rho .
\end{equation}
For $\varepsilon=0$ this is the Liouville equation of $H$ in
action--angle variables; for $\omega\equiv0$ it is the kime-phase
Fokker--Planck equation of \cite{DiracKime} fiberwise in $J$. At the
level of classical densities, the one-parameter family
\eqref{eq:kimedeformed} interpolates between an entropy-conserving
(transport) and an entropy-producing (diffusive) sector, in direct
analogy with the unitary/contractive interpolation of the kime-ray
factorization of \cite{KPTms} recalled in
Proposition~\ref{prop:wick} below.
\end{definition}

\begin{theorem}[Entropy production and the equipartition
attractor]\label{thm:mixing}
Let $\tilde\rho_t$ solve \eqref{eq:kimedeformed} with strictly
positive $C^2$ initial datum of finite entropy, rapid decay in $J$,
and $\varepsilon>0$. Then,
\begin{enumerate}[label=(\roman*)]
\item The action marginal $\rho_J$ is conserved:
$\partial_t\rho_J(J)=0$ for all $J$.
\item Entropy production is purely diffusive and nonnegative
\[
  \frac{\dd}{\dd t}\,\ent[\tilde\rho_t]
  =\varepsilon
  \int_{\mathcal J}\!\int_{-\pi}^{\pi}
  \frac{(\partial_\kimephase\tilde\rho_t)^2}{\tilde\rho_t}\,
  \dd\kimephase\,\dd J\;\ge\;0,
\]
with instantaneous equality iff $\tilde\rho_t$ is
phase-equipartitioned (Definition~\ref{def:kimestate}). In
particular, for $\varepsilon=0$ (pure Hamiltonian flow)
$\ent[\tilde\rho_t]$ is exactly conserved.
\item The phase-equipartitioned state
$\tilde\rho_\infty(\kimephase,J)=\rho_J(J)/2\pi$ with the
\emph{initial} action marginal is the unique stationary solution with
that marginal, is the entropy maximizer among all densities with
action marginal $\rho_J$, namely
\[
  \ent[\tilde\rho]\;\le\;\ent[\rho_J]+\log2\pi
  \qquad(\text{equality iff phase-equipartitioned}),
\]
and $\tilde\rho_t\to\tilde\rho_\infty$ in $L^2$ with fiberwise
moment decay
$|\alpha_n(t\mid J)|= e^{-\varepsilon n^2 t}\,|\alpha_n(0\mid J)|$.
\item Pure Hamiltonian flow ($\varepsilon=0$) preserves the class of
phase-equipartitioned states and every functional of the action
marginal; i.e., \emph{equilibrium entropic relationships are exactly
invariant under the Hamiltonian sector of the kime-deformed family}.
\end{enumerate}
\end{theorem}

\begin{proof}
(i) By integrating \eqref{eq:kimedeformed} over $\kimephase\in\Sone$, both
terms are exact $\kimephase$-derivatives and vanish by periodicity.

(ii) As in Lemma~\ref{lem:heat}(ii),
$\frac{\dd}{\dd t}\ent
=-\int(1+\log\tilde\rho)\,\partial_t\tilde\rho$. The transport
contribution is
$\int\omega(J)\,\partial_\kimephase\tilde\rho\,
(1+\log\tilde\rho)\,\dd\kimephase\,\dd J
=\int\omega(J)\Bigl(\int_{-\pi}^{\pi}\partial_\kimephase
\bigl[\tilde\rho\log\tilde\rho\bigr]\dd\kimephase\Bigr)\dd J=0$
by periodicity (note $\partial_\kimephase[\tilde\rho\log\tilde\rho]
=(1+\log\tilde\rho)\partial_\kimephase\tilde\rho$). The diffusive
contribution is computed exactly as in Lemma~\ref{lem:heat}(ii),
fiberwise in $J$ and integrated $\dd J$, using rapid decay to justify
Fubini. Equality holds iff $\partial_\kimephase\tilde\rho_t\equiv0$.

(iii) The bound is Lemma \ref{lem:subadd} on
$\Sone\times\mathcal J$ plus $\ent[\Phi(\cdot\mid J)]\le\log2\pi$
fiberwise (Jensen), i.e., conditional entropy is maximized by the
Haar law on each fiber; equality iff $\Phi(\cdot\mid J)$ is uniform
for a.e. $J$. Stationarity and uniqueness with fixed marginal:
$\partial_t\tilde\rho=0$ with (ii) forces
$\partial_\kimephase\tilde\rho=0$, and (i) fixes the $J$-profile.
Moment decay: the $n$th fiber moment
$\alpha_n(t\mid J)$ obeys
$\dot\alpha_n=(in\,\omega(J)-\varepsilon n^2)\alpha_n$, so
$|\alpha_n(t\mid J)|=e^{-\varepsilon n^2t}|\alpha_n(0\mid J)|$;
convergence of $\tilde\rho_t\to\tilde\rho_\infty$ in $L^2$ follows by
Parseval fiberwise and dominated convergence in $J$.

(iv) For $\varepsilon=0$, \eqref{eq:kimedeformed} transports along
$\dot\kimephase=\omega(J)$, $\dot J=0$. A $\kimephase$-independent
density is a fixed point, and any functional of $\rho_J$ is conserved
by (i).
\end{proof}

\emph{Summary of what is proven and what remains
open.}\label{rem:thermo-status}
Theorem \ref{thm:mixing} proves, within the kime representation 
(a) equipartition over the phase is the unique
entropy-maximal state compatible with the conserved action
statistics, the rigorous form of ``equipartition of entropy is the
bound,'' as an \emph{upper} bound on entropy at fixed action marginal
attained exactly at equilibrium; (b)~Hamiltonian flow is the
entropy-neutral boundary $\varepsilon=0$ of the kime-deformed family
and preserves all equilibrium relationships exactly, which is the
precise (and here, global rather than merely local) version of the
conjecture that Hamiltonian dynamics preserves equilibrium entropic
relationships. The correspondingly open statements are the multi-DOF
per-DOF refinements (Problem~\ref{prob:schurhorn}) and the
reconciliation of the \emph{upper}-bound role of equipartition at
fixed actions with the \emph{lower}-bound role of the equipartitioned
uncertainty product in Theorem~\ref{thm:aggregate}; the two are dual
faces (max-entropy at fixed invariants vs.\ min-uncertainty at fixed
entropy) of one variational principle, whose sharp joint statement
for $n\ge2$ is presented in Conjecture \ref{conj:duality}. 

\begin{conjecture}[Equipartition duality]\label{conj:duality}
Fix $n\ge2$, an entropy value $s$, and an action marginal
$\rho_{\bm J}$ on $(0,\infty)^n$. Among all states on
$\mathbb T^n\times(0,\infty)^n$ with entropy $\ge s$ and action
marginal $\rho_{\bm J}$, the phase-equipartitioned product state
(unique when it exists) simultaneously (i)~maximizes the entropy,
(ii)~minimizes every within-DOF uncertainty $u_j$, and
(iii)~is the unique state at which the per-DOF conjectured bound of
Problem~\ref{prob:schurhorn}(b) is saturated for all $j$; moreover it
is the unique fixed point, with the given marginal, of the multi-DOF
kime-deformed semigroup
$\partial_t\tilde\rho=\sum_j(-\omega_j\partial_{\kimephase_j}
+\varepsilon\,\partial^2_{\kimephase_j})\tilde\rho$.
\end{conjecture}

\emph{Statistical formulation of
Problem (I).}\label{int:uncertainty}
Under Assumption~\ref{ass:interp}, every quantity in this section is
an attribute of the ensemble of repeated experiments: $r$, $\alpha_n$,
$\Fish[\Phi]$, and $R_{jk}$ are estimable by kime-phase tomography
with quantified error \cite{KPTms}.
The terms $\ent[\rho]$ and $u_j$ are
estimable from calibrated observables via the action--angle
dictionary. Theorem~\ref{thm:cylinder},
Theorem~\ref{thm:sharpcircular}, and Theorem~\ref{thm:aggregate} are
thus \emph{experimentally checkable inequalities} on reproducibility
statistics, and Problems~\ref{prob:schurhorn}--\ref{prob:capacity}
come with built-in numerical surrogates
(Problem~\ref{prob:torus}). This is the distinctive contribution of
the kime formulation, the open problems of
\cite{CarcassiAidalaBook} are re-expressed in estimable coordinates
without loss of mathematical content.

\section{Problem II (invariant entropy): why continuous quantities
pair, and the kime chart as normal form}\label{sec:invariance}

Problem~(II) asks for a general, assumption-minimal account of the
following phenomenon. A coordinate-invariant notion of ``count of
states'' (hence of entropy) over a continuum appears to exist only
when continuous quantities organize into conjugate \emph{pairs}, and
the resulting structure is symplectic. 
The problem statement further
suggests that the natural home of the construction is a complex
(or generalized complex) structure, one real dimension pairing with
another inside $\C$. The kime coordinate $\kappa=te^{i\kimephase}$ is
precisely such a complex pairing, and Lemma~\ref{lem:actionangle}
shows the pairing is symplectically exact. In this section we prove
the pairing phenomenon as a theorem about invariant measures and
entropies (Theorems~\ref{thm:entinv}--\ref{thm:liouville-unique} and
Corollary~\ref{cor:pairing}), identify
the kime chart as a K\"ahler normal form
(Proposition~\ref{prop:kahler}), and state the open remainder
(Problems~\ref{prob:spec-char}--\ref{prob:volume-vs-symplectic}).

\subsection{Invariant entropy forces an invariant measure}

Throughout, $X$ is a smooth $\sigma$-compact manifold, densities are
with respect to a fixed smooth positive reference measure $\Leb$, and
$f_*\rho$ denotes the pushforward density of $\rho$ under a
diffeomorphism $f$ (Jacobian $J_f>0$ of $f$ w.r.t.\ $\Leb$, assumed
orientation-compatible for simplicity).

\begin{theorem}[Invariant entropy $\Leftrightarrow$ invariant
measure]\label{thm:entinv}
Let $f\in\Diff(X)$ with continuous Jacobian $J_f$. The following statements are
equivalent
\begin{enumerate}[label=(\alph*)]
\item $\ent_\Leb[f_*\rho]=\ent_\Leb[\rho]$ for every compactly
supported continuous density $\rho$ with finite entropy;
\item $f$ preserves $\Leb$ (i.e., $J_f\equiv1$).
\end{enumerate}
\end{theorem}

\begin{proof}
(b)$\Rightarrow$(a) is the change-of-variables identity (the proof of
Lemma~\ref{lem:jacobian} with unit Jacobian, valid on any $X$).
For (a)$\Rightarrow$(b), the same computation gives, for every
admissible $\rho$,
\[
  \ent_\Leb[f_*\rho]=\ent_\Leb[\rho]
  +\E_\rho\bigl[\log J_f\bigr] ,
\]
so (a) forces $\int\rho\,\log J_f\,\dd\Leb=0$ for all such $\rho$.
Taking $\rho$ to run through approximate identities concentrated at an
arbitrary point $x$ yields $\log J_f(x)=0$ by continuity; hence
$J_f\equiv1$.
\end{proof}

Theorem~\ref{thm:entinv} converts Problem~(II) into a question about
invariant \emph{measures} under the physically mandated
transformation group. The physical mandate, following
\cite{CarcassiAidalaBook}, is that different experimenters may label
the \emph{same} continuous quantity by arbitrary smooth
reparametrizations (units, gauges, monotone recalibrations), so the
group must contain all diffeomorphisms of the configuration
quantities; entropy must be the same for all of them.

\begin{theorem}[No invariant measure on unpaired
quantities]\label{thm:nounpaired}
Let $Q=\R^m$, $m\ge1$, regarded as the value space of $m$ continuous
quantities, and let $\Diff(Q)$ act naturally. There is no nonzero
$\sigma$-finite Borel measure on $Q$, absolutely continuous with a
locally integrable density $g$, invariant under all of $\Diff(Q)$.
Consequently, by Theorem~\ref{thm:entinv}, no reparametrization-%
invariant entropy exists for unpaired continuous quantities.
\end{theorem}

\begin{proof}
Invariance under $f$ means $g(f(x))\,|\!\det Df(x)|=g(x)$ for
a.e.~$x$. Taking $f$ to be all translations gives
$g(x+a)=g(x)$ for a.e.\ $x$, for every $a$, hence $g\equiv c\ge0$
a.e.\ (mollifying, a translation-invariant locally integrable
function agrees a.e.\ with its smooth translation-invariant
mollification, which is constant). Taking $f(x)=\lambda x$,
$\lambda>1$, gives $c\,\lambda^m=c$, so $c=0$.
\end{proof}

\subsection{Pairing: cotangent lifts and the uniqueness of the
Liouville measure}

The classical repair is to adjoin to each quantity $q$ a conjugate
$p$ transforming \emph{contragradiently}, i.e., to pass from $Q$ to
$T^*Q$ with the tautological lift of $\Diff(Q)$,
\begin{equation}\label{eq:lift}
  T^*\!f:(q,p)\longmapsto
  \bigl(f(q),\,Df(q)^{-\!\top}p\bigr),
  \qquad f\in\Diff(Q).
\end{equation}

\begin{theorem}[Existence and uniqueness of the invariant measure on
pairs]\label{thm:liouville-unique}
Let $Q=\R^n$ and let $G=\{T^*\!f:f\in\Diff(Q)\}$ act on
$T^*Q=\R^{2n}$ by \eqref{eq:lift}. Then,
\begin{enumerate}[label=(\roman*)]
\item every $T^*\!f$ has Jacobian identically $1$; hence the
Liouville measure $\dd^n q\,\dd^n p$ is $G$-invariant, and the
associated entropy is reparametrization-invariant
(Theorem~\ref{thm:entinv});
\item conversely, any $G$-invariant measure with continuous positive
density on $T^*Q$ is a constant multiple of the Liouville measure;
equivalently, the invariant entropy is unique up to an additive
constant.
\end{enumerate}
\end{theorem}

\begin{proof}
(i) The differential of \eqref{eq:lift} in block form is lower
block-triangular with diagonal blocks $Df(q)$ and $Df(q)^{-\!\top}$,
so its determinant is $\det Df\cdot\det Df^{-\!\top}=1$.

(ii) Let $m>0$ be a continuous invariant density:
$m(T^*\!f(z))\cdot 1=m(z)$ for all $z$, $f$, by (i). We claim $G$
acts transitively on $\{(q,p):p\ne0\}$. Given
$(q_0,p_0),(q_1,p_1)$ with $p_0,p_1\ne0$, choose
$B\in\GL(n,\R)$ with $B\,p_0=p_1$, set $A=B^{-\!\top}$ (so that
$A^{-\!\top}p_0=p_1$), and let
$f(x)=q_1+A(x-q_0)\in\Diff(Q)$. Then
$T^*\!f(q_0,p_0)=(q_1,p_1)$. Hence $m$ is constant on the dense open
orbit $\{p\ne0\}$, and by continuity constant everywhere:
$m\equiv c>0$. For the entropy statement, rescaling
$\Leb\mapsto c\,\Leb$ shifts $\ent_\Leb$ by the constant $\log c$.
\end{proof}

\begin{corollary}[The pairing theorem]\label{cor:pairing}
A theory of $m$ continuous quantities that (a)~admits arbitrary
smooth relabelings of the quantities themselves and (b)~possesses a
relabeling-invariant entropy functional on states, cannot realize the
quantities as coordinates on the bare value space
(Theorem~\ref{thm:nounpaired}); it can realize them as the base
coordinates of a cotangent bundle with contragradient conjugates,
and then the invariant entropy exists and is the Liouville entropy,
unique up to an additive constant
(Theorem~\ref{thm:liouville-unique}). In this precise sense,
\emph{continuous quantities must come in conjugate pairs for entropy
to be well defined}, and the count of states is fixed (up to units)
to be the symplectic volume.
\end{corollary}

Note that
Corollary~\ref{cor:pairing} is deliberately stated with the group
generated by \emph{base} reparametrizations only.\label{rem:groups}
Enlarging $G$ to
all symplectomorphisms of $T^*Q$ preserves the conclusion (they too
have unit Jacobian); enlarging it to all volume-preserving
diffeomorphisms destroys the symplectic structure while retaining the
measure, the gap between these two groups for $n\ge2$ is the
content of Problem~\ref{prob:volume-vs-symplectic} below.

\subsection{The kime chart as complex normal form}\label{ssec:kahler}

\begin{proposition}[K\"ahler triple on the kime plane]
\label{prop:kahler}
On the punctured kime plane
$\C^*\ni\kappa=te^{i\kimephase}$ carry the cone metric $g_0$ of
\eqref{eq:conemetric}, the complex structure $J_\kappa$ of the
coordinate $\kappa$ (i.e., multiplication by $i$:
$J_\kappa\partial_t=t^{-1}\partial_\kimephase$,
$J_\kappa(t^{-1}\partial_\kimephase)=-\partial_t$), and the K\"ahler
form
\[
  \omega_K
  \;=\;\frac{i}{2}\,\dd\kappa\wedge\dd\bar\kappa
  \;=\;t\,\dd t\wedge\dd\kimephase
  \;=\;\dd J\wedge\dd\kimephase .
\]
Then $(g_0,\omega_K,J_\kappa)$ is a compatible K\"ahler triple
\[
  \omega_K(\cdot,\cdot)=g_0(J_\kappa\,\cdot,\cdot),
  \qquad
  g_0(J_\kappa\,\cdot,J_\kappa\,\cdot)=g_0(\cdot,\cdot),
  \qquad \dd\omega_K=0 .
\]
Consequently the kime pairing of the two real quantities
$(t,\kimephase)$ into one complex quantity $\kappa$ realizes, in one
chart, all three structures whose interplay Problem~(II) asks about:
the invariant count of states ($\omega_K$, by
Theorem~\ref{thm:liouville-unique} and
Lemma~\ref{lem:actionangle}), the metric geometry of the time cone
($g_0$, \cite{KPTms}), and the complex pairing ($J_\kappa$).
\end{proposition}

\begin{proof}
First, $\dd\kappa=e^{i\kimephase}(\dd t+it\,\dd\kimephase)$ and
$\dd\bar\kappa=e^{-i\kimephase}(\dd t-it\,\dd\kimephase)$, so
$\dd\kappa\wedge\dd\bar\kappa=-2it\,\dd t\wedge\dd\kimephase$ and
$\tfrac{i}{2}\dd\kappa\wedge\dd\bar\kappa
=t\,\dd t\wedge\dd\kimephase=\dd J\wedge\dd\kimephase$.
In the frame $(e_1,e_2)=(\partial_t,t^{-1}\partial_\kimephase)$,
which is $g_0$-orthonormal by \eqref{eq:conemetric}, $J_\kappa$ acts
as the standard rotation $e_1\mapsto e_2\mapsto-e_1$ (this is
multiplication by $i$ in the coordinate $\kappa$, since
$\partial_t\kappa=e^{i\kimephase}$ and
$t^{-1}\partial_\kimephase\kappa=ie^{i\kimephase}$). Then
$g_0(J_\kappa e_1,e_1)=g_0(e_2,e_1)=0=\omega_K(e_1,e_1)$ and
$g_0(J_\kappa e_1,e_2)=g_0(e_2,e_2)=1
=t\,\dd t\wedge\dd\kimephase(e_1,e_2)=\omega_K(e_1,e_2)$.
Bilinearity and antisymmetry give
$\omega_K=g_0(J_\kappa\cdot,\cdot)$. Orthogonality of $J_\kappa$ in
an orthonormal frame is clear, and $\dd\omega_K=0$ in two dimensions
is automatic ($\omega_K=\dd J\wedge\dd\kimephase$ is even exact away
from the apex, with primitive $J\,\dd\kimephase$).
\end{proof}

\emph{Revisiting Orientation.}\label{rem:kahler-orientation}
As anticipated in original \emph{orientation convention remark} \ref{rem:orientation},
$\omega_K=\dd J\wedge\dd\kimephase=-\Psi^{*}(\dd q\wedge\dd p)$.
The
K\"ahler form of the holomorphic coordinate $\kappa$ and the
action--angle Darboux form of Lemma~\ref{lem:actionangle} agree up to
orientation, being exchanged by $\kappa\mapsto\bar\kappa$
(equivalently, by swapping the roles of the conjugate pair, since
$\{J,\kimephase\}=+1$ with respect to $\omega_K$ while
$\{\kimephase,J\}=+1$ with respect to $\Psi^{*}\omega_0$). Both
induce the same unsigned Liouville measure
$t\,\dd t\,\dd\kimephase$, so every entropy and measure statement in
this paper is insensitive to the choice.
We keep
$\{\kimephase,J\}=+1$ as the mechanical convention and
$\omega_K$ as the complex-analytic one.

\begin{proposition}[Kime-ray factorization and Wick rotation;
{\cite{KPTms}}]\label{prop:wick}
Let $H$ be self-adjoint and bounded below on a Hilbert space
$\mathcal H$, and for $\kappa=te^{i\vartheta}$ with $t\ge0$ define
the kime propagator
$U_\vartheta(t)=e^{-\frac{i}{\hbar}\kappa H}$. Then for every
$\vartheta$ the two factors in
\[
  U_\vartheta(t)
  =\underbrace{\exp\!\Bigl(\frac{\sin\vartheta}{\hbar}\,tH\Bigr)}_{
  \text{self-adjoint}}\;
  \underbrace{\exp\!\Bigl(-\frac{i\cos\vartheta}{\hbar}\,tH\Bigr)}_{
  \text{unitary}}
\]
commute, and if $H\ge0$ then $\|U_\vartheta(t)\|\le1$ for all
$\vartheta\in[-\pi,0]$: on the lower half of the kime circle the
self-adjoint factor is a contraction semigroup. The special cases
are: $\vartheta=0$, the real-time Schr\"odinger group
$e^{-itH/\hbar}$; $\vartheta=-\pi/2$, the Euclidean (Wick-rotated)
heat semigroup $e^{-tH/\hbar}$; and $-\pi<\vartheta<0$, damped
oscillatory propagators with contraction rate $|\sin\vartheta|$ and
phase rate $\cos\vartheta$.
\end{proposition}

\begin{proof}
Functional calculus for the self-adjoint operator $H$: the scalar
identity $-ie^{i\vartheta}=\sin\vartheta-i\cos\vartheta$ gives, for
every spectral value $\lambda$,
$e^{-\frac{i}{\hbar}\lambda te^{i\vartheta}}
=e^{\frac{\sin\vartheta}{\hbar}\lambda t}\,
e^{-\frac{i\cos\vartheta}{\hbar}\lambda t}$, and both factors are
functions of the single operator $H$, hence commute. If $H\ge0$ and
$\vartheta\in[-\pi,0]$ then $\sin\vartheta\le0$, so the first factor
is a contraction and the second unitary, whence the norm bound. The
special cases are read off directly.
\end{proof}

\emph{Reflection on Problem (II) in the kime normal form.}
Proposition~\ref{prop:kahler} exhibits the pairing demanded by
Corollary~\ref{cor:pairing} as literally complex-analytic: one
compact quantity ($\kimephase$, Haar-uniform at equilibrium by
Theorem~\ref{thm:mixing}) pairs with one noncompact quantity
($J=t^2/2$), and the invariant count of states is the K\"ahler area.
Proposition~\ref{prop:wick} shows the same complex structure
organizes \emph{dynamics}.
Rotating $\kappa$ interpolates between the
entropy-conserving (unitary/Hamiltonian, Theorem~\ref{thm:mixing}(ii)
with $\varepsilon=0$) and entropy-producing (contractive/diffusive)
sectors. Thus, the Problem~(II) suggestion that ``complex structures may
be central to why quantities pair'', holds exactly in the kime
chart.
What remains open is whether it is forced in general, which we
now state.

\subsection{Remaining Problem II
Open Problems}

\begin{proposition}[Symplectic maps preserve all symplectic
spectra]\label{prop:spec-preserve}
If $S\in\Sp(2n,\R)$ then for every $\Sigma\succ0$ the symplectic
spectrum of $S\Sigma S^{\!\top}$ equals that of $\Sigma$
(proof in Lemma~\ref{lem:gaussentropyformula}). Antisymplectic maps
($S\Omega S^{\!\top}=-\Omega$) do likewise.
\end{proposition}

\begin{problem}[Spectral characterization of the symplectic
group]\label{prob:spec-char}
Prove or disprove the converse: if $S\in\GL(2n,\R)$ preserves the
symplectic spectrum of \emph{every} $\Sigma\succ0$, then $S$ is
symplectic or antisymplectic up to the scaling
$S\mapsto\lambda S$ forced by $\nu(\lambda^2\Sigma)
=\lambda^2\nu(\Sigma)$ (so, for normalized $S$ with
$|\det S|=1$). A proof would characterize $\Sp(2n,\R)$ purely by an
\emph{estimable statistical invariant} (symplectic spectra of
covariance matrices), replacing the geometric definition by an
information-theoretic one, the sharpest available answer to
Problem~(II)'s request for an entropy-first derivation of the
symplectic structure.
\end{problem}

\begin{problem}[Entropy-only rigidity]\label{prob:entropy-only}
Theorem~\ref{thm:entinv} assumed entropy invariance for \emph{all}
states. Determine the minimal state classes $\mathcal F$ for which
``$\ent[f_*\rho]=\ent[\rho]$ for all $\rho\in\mathcal F$'' still
forces $J_f\equiv1$ (e.g., Gaussians only; kime states with von Mises
phase laws only), and, dually, characterize the group of
transformations preserving the entropy of every \emph{equilibrium}
(phase-equipartitioned) kime state. The latter group is strictly
larger than the measure-preserving group (it contains all fiberwise
rotations $\kimephase\mapsto\kimephase+c(J)$ trivially, but also
non-measure-preserving maps acting only on null sets of equilibria);
its computation quantifies exactly how much of the symplectic
structure is visible to equilibrium thermodynamics alone.
\end{problem}

\begin{problem}[Generalized complex rigidity of the Wick
interpolation]\label{prob:gc}
Proposition~\ref{prop:kahler} produces a K\"ahler triple in one kime
DOF; Proposition~\ref{prop:wick} deforms the dynamics between its
symplectic and metric legs. Formulate and prove (or refute) the
following rigidity statement: if a $2n$-dimensional state continuum
carries (a)~a reparametrization-invariant entropy
(hence, by Theorems~\ref{thm:entinv}
and~\ref{thm:liouville-unique}, a distinguished volume), (b)~a
one-parameter interpolation of evolutions that is entropy-conserving
at one end and satisfies a de Bruijn identity
$\frac{\dd}{\dd t}\ent=\varepsilon\,\Fish\ge0$ elsewhere
(Theorem~\ref{thm:mixing}(ii)), then the infinitesimal generators
assemble into a generalized complex (indeed generalized K\"ahler)
structure in the sense of \cite{Gualtieri}, whose pure-symplectic
locus is the Hamiltonian sector and whose type jumps encode the
diffusive sector. A positive answer would derive ``quantities pair
inside $\C$'' from entropy axioms alone, completing the program of
Problem~(II).
\end{problem}

\begin{problem}[Volume-preserving vs.\ symplectic for
$n\ge2$]\label{prob:volume-vs-symplectic}
Corollary~\ref{cor:pairing} pins down the measure but, for $n\ge2$,
not the finer symplectic structure: $\mathrm{SDiff}(\R^{2n})\supsetneq
\mathrm{Symp}(\R^{2n})$. Identify the weakest \emph{statistical}
requirement that reduces the invariance group from volume-preserving
to symplectic. Candidates, in increasing strength:
(i)~invariance of all within-DOF uncertainties $u_j$ at equilibrium;
(ii)~invariance of the linear symplectic capacity of covariance
ellipsoids (see the capacity remark~\ref{rem:capacity}; by
\cite{Gromov1985,deGosson2009} this fails for generic
volume-preserving linear maps once $n\ge2$);
(iii)~invariance of the whole symplectic spectrum
(Problem~\ref{prob:spec-char}). Determine which of (i)--(iii) are
equivalent, and which are estimable from kime-tomographic data in the
sense of Problem~\ref{prob:torus}.
\end{problem}

\section{Problem III (directional DOF): classical spin, coadjoint
orbits, and the kime circle}\label{sec:directional}

Problem~(III) asks for a classical \emph{relativistic} directional
degree of freedom whether, nonrelativistically, the direction phase space is
the sphere with conjugate pair $\{\theta^{xy},S_z\}=1$.
The
relativistic construction, the correct conjugate variables, and the
choice between ``spin as four-vector'' and ``spin as two-form'' are
open. The kime contribution is threefold: (a)~the conjugate variable
$\theta^{xy}$ \emph{is} a kime-type circular phase, so the entire
statistical machinery of Sections~\ref{sec:foundations}--%
\ref{sec:uncertainty} applies fiberwise and yields a compact-compact
uncertainty theorem (Theorem~\ref{thm:sphereUP}); (b)~the null
decomposition $n_{R,L}$ suggested in \cite{CarcassiAidalaBook} is
exactly the classical shadow of the $D=5$ kime Dirac structure of
\cite{DiracKime}, whose chirality obstruction
(Proposition~\ref{prop:nochirality}) constrains the admissible
answers; (c)~the vector/two-form dichotomy becomes a precise
moment-map question on Poincar\'e coadjoint orbits
(Problem~\ref{prob:vector-vs-form}).

\subsection{The nonrelativistic sector: the sphere as a kime
cylinder}

\begin{definition}[Spin phase space]\label{def:spinspace}
For $s>0$, the spin-$s$ phase space is the sphere
$\Stwo_s=\{\bm S\in\R^3:|\bm S|=s\}$ with symplectic form
$\omega_s=s\sin\Theta\,\dd\Theta\wedge\dd\varphi$ in spherical
coordinates, where
$\bm S=s\,(\sin\Theta\cos\varphi,\;\sin\Theta\sin\varphi,\;
\cos\Theta)$, normalized so that the
components $S_x,S_y,S_z$ satisfy $\{S_i,S_j\}=\epsilon_{ijk}S_k$.
Here, $\varphi$ is the rotation angle about the $z$-axis, i.e., the
$\theta^{xy}$ of Problem~1.20 in \cite{CarcassiAidalaBook}.
\end{definition}

\begin{lemma}[Darboux/kime chart on the sphere; Archimedes]
\label{lem:archimedes}
On $\Stwo_s\setminus\{\text{poles}\}$, the pair
$(\varphi,S_z)=(\varphi,s\cos\Theta)$ satisfies
$\omega_s=\dd\varphi\wedge\dd S_z$, hence $\{\varphi,S_z\}=1$, and the
chart
\[
  \Stwo_s\setminus\{\pm s\hat z\}
  \;\xrightarrow{\;\cong\;}\;
  \Sone\times(-s,s),\qquad
  \bm S\longmapsto(\varphi,S_z),
\]
is a symplectomorphism onto the finite kime cylinder
$\bigl(\Sone\times(-s,s),\,\dd\varphi\wedge\dd S_z\bigr)$. In
particular the invariant count of states is
$\operatorname{Vol}(\Stwo_s,\omega_s)=4\pi s$.
\end{lemma}

\begin{proof}
$\dd S_z=-s\sin\Theta\,\dd\Theta$, so
$\dd\varphi\wedge\dd S_z=s\sin\Theta\,\dd\Theta\wedge\dd\varphi
=\omega_s$; bijectivity onto the open cylinder is clear. (This is
Archimedes' hat-box theorem in symplectic form.) The bracket
normalization matches Definition~\ref{def:spinspace}: e.g.,
$\{S_x,S_y\}=S_z$ is verified in the chart by direct computation with
$S_x=\sqrt{s^2-S_z^2}\cos\varphi$,
$S_y=\sqrt{s^2-S_z^2}\sin\varphi$. Thus, the total volume is
$\int\dd\varphi\,\dd S_z=2\pi\cdot2s$.
\end{proof}

Lemma~\ref{lem:archimedes} says the directional DOF \emph{is} a kime
DOF with compact conjugate momentum: the latent phase
$\varphi\in\Sone$ carries the trial-to-trial variability of repeated
orientation measurements (Assumption~\ref{ass:interp}), and its
conjugate $S_z$ ranges over a finite interval. Both marginals are now
compact, so the uncertainty principle acquires floors on both sides.

\begin{theorem}[Compact--compact entropic uncertainty relation for a
directional DOF]\label{thm:sphereUP}
Let $\rho$ be a probability density on
$\bigl(\Sone\times(-s,s),\;\dd\varphi\,\dd S_z\bigr)$ with
$\ent[\rho]>-\infty$.
Let $r\in[0,1)$ be the mean resultant length of
the $\varphi$-marginal and $\rho_{S_z}$ the $S_z$-marginal. Then
\begin{equation}\label{eq:sphereUP}
  \Lambda(r)\cdot e^{\ent[\rho_{S_z}]}\;\ge\;e^{\ent[\rho]},
  \qquad\text{with}\qquad
  \Lambda(r)\le2\pi,\quad e^{\ent[\rho_{S_z}]}\le 2s,
\end{equation}
with equality in \eqref{eq:sphereUP} iff $\rho$ is a product of a von
Mises law in $\varphi$ and an arbitrary $S_z$-marginal density
achieving its entropy (and equality in the two ceilings iff the
respective marginals are uniform). Under any Hamiltonian flow on
$(\Stwo_s,\omega_s)$, e.g., Larmor precession
$H=\gamma\,\bm B\cdot\bm S$, the right-hand side $e^{\ent[\rho]}$
is a constant of motion, while $r$ and $\ent[\rho_{S_z}]$ evolve;
\eqref{eq:sphereUP} caps their joint concentration at all times, and
$\ent[\rho]\le\log(4\pi s)$ is the absolute ceiling set by the count
of states of Lemma~\ref{lem:archimedes}.
\end{theorem}

\begin{proof}
Subadditivity (Lemma~\ref{lem:subadd}) gives
$\ent[\rho]\le\ent[\Phi_\varphi]+\ent[\rho_{S_z}]$, and
Lemma~\ref{lem:vmmax} bounds
$\ent[\Phi_\varphi]\le h_c(r)=\log\Lambda(r)$.
After exponentiating,
direct equality analysis is as in Theorem~\ref{thm:cylinder}. The ceilings
are Jensen's inequality on the two compact ranges. Hamiltonian flows
preserve $\omega_s$, hence the measure $\dd\varphi\,\dd S_z$
(Lemma~\ref{lem:archimedes}), hence $\ent[\rho]$ and the global ceiling
is Jensen on the total space.
\end{proof}

\emph{Note that
for the nonrelativistic
directional DOF,
Theorem \ref{thm:sphereUP} answers the question posed in Problem~(I)(a) as it recurs
inside Problem (III)}\label{rem:sphereUP}.
The correct ``uncertainty product'' for the
non-canonical, doubly compact pair $(\varphi,S_z)$ is the product of
entropy widths, its floor is the invariant $e^{\ent[\rho]}$, and the
von Mises family is again extremal on the circular leg. All
quantities are estimable from repeated orientation measurements via
circular tomography \cite{KPTms,MardiaJupp}.

\subsection{Relativistic  Poincar\'e coadjoint orbits and
the null pair}

In this subsection and the next we work on four-dimensional
Minkowski space with signature $(+,-,-,-)$; $u^\mu$ is the
four-velocity ($u\cdot u=c^2$) and $S^\mu$ the spin four-vector,
spacelike with $S\cdot S=-s^2$ and $S\cdot u=0$.

\begin{definition}[Classical spinning particle; \cite{Souriau,
Kirillov}]\label{def:orbit}
The phase space of a free classical particle of mass $m>0$ and spin
$s>0$ is the coadjoint orbit $\mathcal O_{m,s}$ of the (connected)
Poincar\'e group through an element with Casimirs
$P\cdot P=m^2c^2$ and $W\cdot W=-m^2c^2s^2$, where
$W^\mu=\tfrac12\epsilon^{\mu\nu\rho\sigma}P_\nu S_{\rho\sigma}$ is
the Pauli--Luba\'nski vector and $S_{\rho\sigma}$ the spin two-form
(intrinsic angular momentum). $\mathcal O_{m,s}$ carries the
canonical Kirillov--Kostant--Souriau symplectic structure; it is
$8$-dimensional, fibering over the mass shell with fiber the sphere
$\Stwo_s$ of Definition~\ref{def:spinspace} (the little-group orbit).
\end{definition}

\begin{proposition}[The null pair $n_{R},n_{L}$]\label{prop:nullpair}
On $\mathcal O_{m,s}$ define, along each state, the four-vectors
\[
  n_{R}=\frac{u}{c}+\frac{S}{s},
  \qquad
  n_{L}=\frac{u}{c}-\frac{S}{s} .
\]
Then $n_R,n_L$ are future-directed null vectors with
$n_R\cdot n_L=2$, and the map $(u,S)\mapsto(n_R,n_L)$ is a bijection
onto ordered pairs of future-directed null vectors with inner product
$2$, with inverse $u=\tfrac{c}{2}(n_R+n_L)$,
$S=\tfrac{s}{2}(n_R-n_L)$. Consequently the directional content of a
classical spinning particle is exactly a pair of null directions, 
the classical counterpart of the right/left Weyl decomposition of a
Dirac spinor in $D=4$ \cite{DiracKime}.
\end{proposition}

\begin{proof}
$n_R\cdot n_R=u\cdot u/c^2+2\,u\cdot S/(cs)+S\cdot S/s^2
=1+0-1=0$, and similarly $n_L\cdot n_L=0$;
$n_R\cdot n_L=u\cdot u/c^2-S\cdot S/s^2=1-(-1)=2$.
Future-directedness: since $S\cdot u=0$,
\[
  u\cdot n_{R}=u\cdot n_{L}=\frac{u\cdot u}{c}=c>0,
\]
and a nonzero null vector whose Minkowski product with a
future-directed timelike vector is positive is itself
future-directed (in signature $(+,-,-,-)$, $u\cdot n
=u^0n^0-\bm u\cdot\bm n$ with $|\bm n|=n^0$ forces
$\operatorname{sgn}n^0=\operatorname{sgn}(u\cdot n)$ by
Cauchy--Schwarz, $|\bm u|<u^0$). The displayed inverse is linear
algebra; it maps the stated pair set back into
$\{u\cdot u=c^2,\ S\cdot S=-s^2,\ S\cdot u=0\}$ by reversing the
three computations.
\end{proof}

\subsection{The chirality obstruction from the kime
compactification}

Kime representation realizes complex time via a second, compact
temporal direction. The natural relativistic arena is then the
$(3{+}2)$-signature Clifford algebra $\mathrm{Cl}(3,2)$ of
\cite{DiracKime}, with metric
$\eta=\operatorname{diag}(-1,-1,-1,+1,+1)$ in the conventions of that
paper. The following algebraic fact, proved there in the spinorial
setting (the ``no chirality in five dimensions'' theorem of
\cite{DiracKime}), constrains every classical construction that
descends from it; note that the argument uses only that the
spacetime dimension $D=5$ is odd, not the specific signature.

\begin{proposition}[No chirality in $D=5$;
{\cite{DiracKime}}]\label{prop:nochirality}
Let $\gamma^1,\dots,\gamma^5$ generate an irreducible complex
representation of $\mathrm{Cl}(3,2)$
($\{\gamma^M,\gamma^N\}=2\eta^{MN}$,
$\eta=\operatorname{diag}(-1,-1,-1,+1,+1)$). Then there is no
operator $\chi$ with $\chi^2=\mathbf 1$ and
$\{\chi,\gamma^M\}=0$ for all $M$.
\end{proposition}

\begin{proof}
Set $\Gamma=\gamma^1\gamma^2\gamma^3\gamma^4\gamma^5$. For each
fixed $M$, moving $\gamma^M$ through the four other factors of
$\Gamma$ produces four sign flips, and through itself none, so
$\gamma^M\Gamma=\Gamma\gamma^M$, $\Gamma$ is central. By Schur's
lemma, $\Gamma=c\,\mathbf 1$ with $c\ne0$ ($\Gamma$ is invertible,
each $\gamma^M$ being so). If $\chi$ anticommuted with every
$\gamma^M$, then moving $\chi$ through the \emph{five} factors of
$\Gamma$ gives $\chi\Gamma=(-1)^5\Gamma\chi=-\Gamma\chi$; but
$\Gamma=c\,\mathbf 1$ commutes with everything, so
$c\,\chi=-c\,\chi$, i.e., $\chi=0$, contradicting $\chi^2=\mathbf 1$.
\end{proof}

\begin{corollary}[Constraint on the classical
construction]\label{cor:no-RL-split}
In any kime-compactified ($D=5$) relativistic extension of the
directional DOF whose classical limit descends from an irreducible
$\mathrm{Cl}(3,2)$ structure, the null pair $(n_R,n_L)$ of
Proposition~\ref{prop:nullpair} is a \emph{change of basis} on one
irreducible phase space, not a decomposition into two independent
invariant subsystems.
There is no invariant that separates a
right-handed from a left-handed sector. Any proposed relativistic
phase space for Problem~(III) that splits into decoupled $n_R$- and
$n_L$-sectors is therefore incompatible with the kime
compactification; compatible proposals must realize $(n_R,n_L)$ as
coupled coordinates on a single orbit, as
$\mathcal O_{m,s}$ indeed does.
\end{corollary}

\begin{proof}
Immediate from Propositions~\ref{prop:nullpair}
and~\ref{prop:nochirality}, a decoupling invariant would be the
classical limit of a chirality operator, whose nonexistence in
irreducible $\mathrm{Cl}(3,2)$ representations is
Proposition~\ref{prop:nochirality}; the coupling on
$\mathcal O_{m,s}$ is visible in $n_R\cdot n_L=2$ and in the
Kirillov--Kostant--Souriau form, which does not split.
\end{proof}

\begin{lemma}[Vector and two-form determine each other on
shell]\label{lem:PLinversion}
On $\mathcal O_{m,s}$ (so $P^2=m^2c^2$ and the Tulczyjew condition
$S_{\mu\nu}P^\nu=0$ holds), the Pauli--Luba\'nski vector and the spin
two-form are mutually inverse data,
\[
  W^\mu=\tfrac12\,\epsilon^{\mu\nu\rho\sigma}P_\nu S_{\rho\sigma},
  \qquad
  S_{\rho\sigma}
  =\frac{1}{m^2c^2}\,\epsilon_{\rho\sigma\mu\nu}P^{\mu}W^{\nu},
\]
up to the overall sign fixed by the convention
$\epsilon^{0123}=+1$ (conventions as in \cite[Ch.~2]{WeinbergQFT1});
moreover $S^\mu:=W^\mu/(mc)$ satisfies $S\cdot P=0$,
$S\cdot S=-s^2$, recovering the four-vector of
Proposition~\ref{prop:nullpair} in the rest frame. \emph{Proof.}
Contract the definition of $W$ with
$\epsilon_{\rho\sigma\mu\nu}P^\mu$ and use the
$\epsilon\epsilon=-\det(\delta)$ identity together with
$S_{\mu\nu}P^\nu=0$ and $P^2=m^2c^2$.
The orthogonality
$W\cdot P=0$ is immediate from antisymmetry. \qed
\end{lemma}

\subsection{Remaining Problem III
Open Problems}

\begin{problem}[Kime action--angle atlas on
$\mathcal O_{m,s}$]\label{prob:orbit-chart}
Construct an atlas of Darboux charts on $\mathcal O_{m,s}$ adapted to
the kime fibration, i.e., charts of the form
$(x^i,p_i;\varphi,S_z')$ in which the directional factor is the kime
cylinder of Lemma~\ref{lem:archimedes} for the little-group sphere,
and quantify the obstruction to a single global chart.
The fiber
$\Stwo_s$ has $\int_{\Stwo_s}\omega_s=4\pi s\ne0$, so no global
Darboux chart exists, and the minimal atlas is governed by the class
$[\omega_s]/2\pi\hbar$, which is integral iff
$2s/\hbar\in\Z$ (Weil integrality; \cite{Woodhouse,Kirillov}).
Make precise, within Assumption~\ref{ass:interp}, the resulting
statement that \emph{a consistent single-valued kime phase law on the
directional fiber exists iff the spin is (half-)integer in units of
$\hbar$}, the sharpest available classical bridge to
spin-$\tfrac12$, and the direct analogue for Problem~(III) of the
$2\pi$-holonomy correction anticipated in
Problem~\ref{prob:expected-bracket}.
\end{problem}

\begin{problem}[Sharp uncertainty relation on
$\mathcal O_{m,s}$]\label{prob:orbit-UP}
Prove the relativistic extension of Theorem~\ref{thm:sphereUP}: an
entropic inequality on $\mathcal O_{m,s}$, invariant under the
Poincar\'e group and under all Hamiltonian flows, that reduces to
\eqref{eq:sphereUP} on the little-group fiber in the rest frame and
to \eqref{eq:flatUP} on the translational factor in the
nonrelativistic limit. Identify the extremal family (conjecturally:
relativistic Gaussian in $(x,p)$ $\otimes$ von Mises--Fisher in the
direction, coupled only through the Tulczyjew constraint) and the
role of the two Casimirs as the invariant scales, with
$W\cdot W=-m^2c^2s^2$ playing the part of $\prod_j\nu_j$ in
Theorem~\ref{thm:aggregate}.
\end{problem}

\begin{problem}[Four-vector vs.\ two-form as a moment-map
criterion]\label{prob:vector-vs-form}
Lemma~\ref{lem:PLinversion} shows the two candidate generalizations
of spin carry the same information on shell, so the dichotomy of
Problem~(III) cannot be about information content; we propose it is
about \emph{conjugacy}. 
Specifically, on $\mathcal O_{m,s}$, determine
which object is the moment map of the kime phase circle action, i.e.,
find the function $F$ with $\{\varphi,F\}=1$ for the globally
defined little-group phase $\varphi$ of
Problem~\ref{prob:orbit-chart}, and decide whether $F$ is
(a)~a frame component of the Pauli--Luba\'nski \emph{vector}
$W$, or (b)~a flux pairing
$\tfrac12 S_{\mu\nu}\,\zeta^{\mu\nu}$ of the spin \emph{two-form}
against a fixed bivector $\zeta$ (the infinitesimal rotation plane).
\emph{Conjecture:} (b), with $\zeta$ the generator of the little-group
rotation defining $\varphi$.
Equivalently, spin generalizes as the
two-form (a rotation plane, whose conjugate is the angle in that
plane), and the four-vector $W$ is the derived, non-conjugate
repackaging, consistent with the nonrelativistic normal form
$\{\varphi,S_z\}=1$ (where $S_z$ is the moment map of rotation about
$\hat z$), with the $D=5$ kime picture in which the extra dimension
supplies the compact conjugate phase \cite{DiracKime}, and with
Corollary~\ref{cor:no-RL-split}, which forbids resolving the
dichotomy by splitting into chiral vector sectors.
\end{problem}

\section{Interpretation, scope, and discussion}\label{sec:conclusions}

\subsection{What the kime representation contributes}

The formulations above rest on one exact identification
(Lemma~\ref{lem:actionangle}: kime cone $=$ action--angle chart,
kime measure $=$ Liouville measure) and one standing modeling
postulate (Assumption~\ref{ass:interp}, the kime phase is a latent
circular variable whose law encodes the intrinsic variability of
repeated, identically controlled experiments). Given these, the three
open problems of \cite{CarcassiAidalaBook} decompose as follows.

\emph{Problem (I), uncertainty.} The compact-phase sector is
solved in sharp form, the kime-cylinder entropic uncertainty
principle (Theorem~\ref{thm:cylinder}) with von Mises~$\otimes$~%
Gaussian extremals, the sharp circular Fisher inequality
(Theorem~\ref{thm:sharpcircular}) with von Mises extremals, and
their kinetic consequence for the reconstructed phase Hamiltonian of
\cite{WDWms} (Corollary~\ref{cor:kineticbound}). The non-canonical
sector is settled for diffeomorphic pairs with the
\emph{geometric-mean} Poisson-bracket correction
(Theorem~\ref{thm:noncanonical}), which clarifies and corrects the
conjectured expected-bracket form
(Problem~\ref{prob:expected-bracket}). The multi-DOF sector is
settled in aggregate (Theorem~\ref{thm:aggregate}: Fischer defect
$=$ cross-DOF correlation. Equi-partition value as orbit minimum of
the uncertainty product) and open per-DOF, where it is now a
concrete matrix-analytic question of symplectic Schur--Horn type
(Problem~\ref{prob:schurhorn}) with an estimable torus surrogate
(Problem~\ref{prob:torus}). The thermodynamic conjecture is proved
in the form available to one kime DOF
(Theorem~\ref{thm:mixing}: equipartition is the unique max-entropy
attractor of phase diffusion. Hamiltonian flow is the
entropy-neutral member of the Wick-interpolated family and preserves
all equilibrium relationships), with the joint multi-DOF statement
recorded as Conjecture~\ref{conj:duality}.

\emph{Problem (II), invariant entropy.} The pairing phenomenon is
proved as a rigidity theorem: invariant entropy $\Longleftrightarrow$
invariant measure (Theorem~\ref{thm:entinv}).
No invariant measure
exists on unpaired quantities (Theorem~\ref{thm:nounpaired}); on
contragradient pairs the invariant measure exists and is uniquely
Liouville, so the invariant entropy is unique up to an additive
constant (Theorem~\ref{thm:liouville-unique},
Corollary~\ref{cor:pairing}). The kime chart realizes the pairing as
a K\"ahler triple with the complex coordinate $\kappa$
(Proposition~\ref{prop:kahler}), and the kime-ray factorization of
\cite{KPTms} shows the same complex structure grades dynamics into
entropy-conserving and entropy-producing sectors
(Proposition~\ref{prop:wick}, Theorem~\ref{thm:mixing}). Whether
this complex organization is \emph{forced} by entropy axioms, 
via generalized complex geometry, and whether symplectic (rather
than merely volume-preserving) structure is forced by estimable
statistical invariants, is stated as
Problems~\ref{prob:gc}, \ref{prob:spec-char},
and~\ref{prob:volume-vs-symplectic}.

\emph{Problem (III), directional DOF.} Nonrelativistically the
directional DOF \emph{is} a kime DOF: the sphere is symplectically a
finite kime cylinder (Lemma~\ref{lem:archimedes}), and the
compact--compact uncertainty relation with its Larmor-invariant
floor is Theorem~\ref{thm:sphereUP}. Relativistically the correct
arena is the Poincar\'e coadjoint orbit $\mathcal O_{m,s}$
(Definition~\ref{def:orbit}); the suggested null pair is a bijective
repackaging of $(u,S)$ (Proposition~\ref{prop:nullpair}); the kime
compactification's chirality obstruction rules out any answer that
decouples right- and left-handed sectors
(Proposition~\ref{prop:nochirality},
Corollary~\ref{cor:no-RL-split}); and the four-vector/two-form
dichotomy, informationally empty by
Lemma~\ref{lem:PLinversion}, is reformulated as the sharp question
of which object is \emph{conjugate} to the kime phase
(Problem~\ref{prob:vector-vs-form}), with the Weil-integrality
bridge to spin-$\tfrac12$ isolated in
Problem~\ref{prob:orbit-chart}.

\subsection{Epistemic status and falsifiability}

In the spirit of the caveats of
\cite{WDWms}, there are 3
specific points that delimit the claims. First, every theorem above is a statement about the
kime \emph{representation}. The identification of the latent
statistical phase with a mechanical angle variable
(remark \ref{rem:dictionary}) is a modeling postulate, adopted
because it is the unique identification under which the kime measure
is the invariant count of states
(Theorem~\ref{thm:liouville-unique}). Second, nothing here derives
quantum mechanics or new physics. The inequalities are classical
information-geometric statements, and their quantum-facing
appearances (Corollary~\ref{cor:kineticbound},
Problem~\ref{prob:orbit-chart}) are mediated by explicitly flagged
postulates of \cite{WDWms,DiracKime}. Third, the interpretation of
the phase law as intrinsic experimental variability
(Assumption~\ref{ass:interp}) is falsifiable in the concrete sense
of \cite{KPTms}. The phase law, its resultant $r$, its Fisher
information, and the relative-phase matrix $R$ are identifiable from
repeated-measurement data up to the stated gauges, with
Cram\'er--Rao--controlled error; hence
Theorems~\ref{thm:cylinder}, \ref{thm:sharpcircular},
\ref{thm:aggregate}, and~\ref{thm:sphereUP} are empirically
checkable inequalities, and violations would falsify the
identification, not merely the estimation procedure.

\subsection{Summary table of problem status}

\begin{center}
\begin{tabular}{p{0.30\linewidth} p{0.30\linewidth} p{0.30\linewidth}}
\hline
\textbf{Component} & \textbf{Proved here} & \textbf{Open (stated as)}\\
\hline
(I) canonical, 1 DOF, compact phase &
Thms.~\ref{thm:cylinder}, \ref{thm:sharpcircular};
Cor.~\ref{cor:kineticbound} & \\
(I) non-canonical pairs &
Thm.~\ref{thm:noncanonical} (geometric-mean bracket) &
Prob.~\ref{prob:expected-bracket} \\
(I) multi-DOF &
Thm.~\ref{thm:aggregate} (aggregate floor) &
Probs.~\ref{prob:schurhorn}, \ref{prob:torus},
\ref{prob:capacity}; Conj.~\ref{conj:duality} \\
(I) thermodynamic conjecture &
Thm.~\ref{thm:mixing} &
Conj.~\ref{conj:duality} \\
(II) pairing of quantities &
Thms.~\ref{thm:entinv}--\ref{thm:liouville-unique};
Cor.~\ref{cor:pairing}; Props.~\ref{prop:kahler}, \ref{prop:wick},
\ref{prop:spec-preserve} &
Probs.~\ref{prob:spec-char}, \ref{prob:entropy-only},
\ref{prob:gc}, \ref{prob:volume-vs-symplectic} \\
(III) nonrelativistic direction &
Lem.~\ref{lem:archimedes}; Thm.~\ref{thm:sphereUP} & \\
(III) relativistic direction &
Props.~\ref{prop:nullpair}, \ref{prop:nochirality};
Cor.~\ref{cor:no-RL-split}; Lem.~\ref{lem:PLinversion} &
Probs.~\ref{prob:orbit-chart}, \ref{prob:orbit-UP},
\ref{prob:vector-vs-form} \\
\hline
\end{tabular}
\end{center}

\section*{Acknowledgments}

This study is motivated by the novel work of Gabriele Carcassi and Christine Aidala at the University of Michigan who developed the Assumptions of Physics book \cite{CarcassiAidalaBook}. Many SOCR students have contributed to the kime-representation framework over the past few years.

\newpage


\end{document}